\documentclass[hidelinks]{article}
\usepackage{amsmath,amssymb,amsthm,bbm,theoremref}
\usepackage{booktabs, csquotes, graphicx, subfig, siunitx, mathtools, xfrac, multirow}
\usepackage{bookmark, ragged2e, caption, float}
\usepackage[table]{xcolor}
\usepackage[linesnumbered,ruled,vlined]{algorithm2e}

\colorlet{mygray}{gray!10!white}

\graphicspath{ {./Figures/} }

\def\tos{\rightrightarrows}
\def\reals{\mathbbm{R}}
\def\naturals{\mathbbm{N}}
\def\cl{\mathop{\rm cl}}
\def\co{\mathop{\rm co}}
\def\ri{\mathop{\rm ri}}
\def\inte{\mathop{\rm int}}

\def\supp{\mathop{\rm supp}\nolimits}
\def\F{{\cal F}}
\def\N{{\cal N}}
\def\FF{(\F_t)_{t=0}^T}
\def\reals{\mathbbm{R}}
\def\minimize{\mathop{\rm minimize}\limits}
\def\ovr{\mathop{\rm over}\ }
\def\st{\mathop{\rm subject\ to}}

\newtheorem{theorem}{Theorem}

\newtheorem{example}[theorem]{Example}

\newtheorem{remark}[theorem]{Remark}

\title{Statistical modeling of SOFR term structure}
\author{Teemu Pennanen\thanks{Department of Mathematics, King's College London, Strand, London, WC2R 2LS, United Kingdom, teemu.pennanen@kcl.ac.uk} \and Waleed Taoum\thanks{Department of Mathematics, King's College London, Strand, London, WC2R 2LS, United Kingdom, waleed.taoum@kcl.ac.uk}}
\date{November 4, 2025}

\begin{document}

\pdfbookmark{Table of Contents}{name} 

\maketitle

\begin{abstract}
SOFR derivatives market remains illiquid and incomplete, so it is not amenable to classical risk-neutral term structure models which are based on the assumption of perfect liquidity and completeness. This paper develops a statistical SOFR term structure model that is well-suited for use in risk management and derivatives pricing models within the incomplete markets paradigm. The model incorporates relevant macroeconomic factors that drive central bank policy rates which, in turn, cause jumps often observed in the SOFR rates. The model is easy to calibrate to historical data, current market quotes, and the user's views concerning the future development of the relevant macroeconomic factors. The model is well suited for large-scale simulations often required in risk management, portfolio optimization and indifference pricing of interest rate derivatives.
\end{abstract}

\section{Introduction}

SOFR, the new benchmark interest rate in the U.S., is an overnight rate computed daily based on repurchase agreement (repo) transactions. The term structure of interest rates is largely captured by exchange-traded futures contracts on SOFR averages over three- and one-month periods. In addition to the futures contracts, there are currently around 2,000 exchange-traded options on the futures rates for 20 different maturities. While the SOFR derivatives market has grown rapidly over the past years, it still exhibits significant bid-ask spreads and limited liquidity. The liquidity is particularly low for many of the listed options, which often have zero liquidity. 

Due to the illiquidity and incompleteness, classical arbitrage arguments are not well suited for pricing SOFR derivatives. Indeed, risk-neutral expectations yield linear pricing rules that are poor descriptions of observed market prices or the prices at which an agent might be willing to trade. In the face of incompleteness, the replication arguments behind risk-neutral valuations do not hold. In incomplete and illiquid markets, offered prices are best described by the indifference pricing principle, which is based on approximate hedging arguments capturing nonlinearities and the subjectivity of offered prices. Indifference pricing builds on a subjective (as opposed to risk-neutral) probabilistic description of the underlying rates, prices, and other relevant risk factors such as inflation and GDP; see e.g.~\cite{car9,pen14}.

This paper presents a statistical model that describes the development of the SOFR term structure under the {\em subjective measure}, also known as the $P$-measure, real-world measure, or statistical measure. Our model captures many of the statistical features observed in the market. Due to the Federal Reserve (FED) active participation in the repo market, SOFR tends to closely track the FED's policy rates. This causes jumps in SOFR whenever the FED changes its policy rates at the Federal Open Market Committee (FOMC) meeting dates. The FED rates are, in turn, largely driven by macroeconomic factors such as inflation and GDP. On the other hand, the macroeconomic factors are strongly influenced by the FED rates. While the SOFR is calculated at a daily frequency, the FED rates and the macroeconomic variables are updated on much lower frequencies. The FOMC meets roughly 8 times a year, while inflation and GDP estimates are published on a monthly frequency.

The model presented in this paper captures all the above features. The model includes a description of the development of macroeconomic variables, including the FED rate, on the FOMC meeting dates, which occur roughly on a monthly frequency. SOFR rate and the associated forward curve are then modelled in relation to the FED rate on a daily frequency. This model structure reflects the actual dependence of SOFR on the FED rates, and it captures the jumps observed in SOFR whenever the Fed decides to change the rates. The model is easy to calibrate to historical data, user's views and observed futures rates. It is fast in simulations, which makes it well-suited for both pricing and risk management applications. We illustrate the simulation model with indifference pricing of SOFR futures, options and swaptions. The pricing exercises here are meant only as an illustration of the convenience of the model in large-scale simulations. The pricing and hedging applications will be developed further in a follow-up paper, that builds on computational optimization of semi-static hedging strategies employing hundreds of exchange-traded SOFR derivatives.

Term structure models under $P$-measure are widely used by public organisations and asset managers. Many of them are based on the Nelson-Siegel and Svensson parameterizations, which describe the yield curve with three or four parameters, respectively; see \cite{nelson1987parsimonious, svensson1994estimating}. For example, \cite{diebold2006macroeconomy} develops a stochastic term structure model where the yield curve is statistically connected to macroeconomic variables. An extension with stochastic volatility was proposed in~\cite{bianchi2009great} while \cite{yu2011forecasting} incorporated corporate bonds. The model presented in this paper allows for more flexible parameterizations, and it incorporates jumps on FOMC meeting dates, which is an essential feature of the overnight lending market.

In derivative pricing, term structure models are often specified so that the corresponding zero-coupon bond prices have arbitrage-free dynamics. A common approach is to specify the term structure dynamics together with a risk-neutral measure (or, more generally, a pricing kernel) under which zero-coupon bond prices of all maturities are martingales. Risk-neutral term structure models for SOFR have been developed in \cite{bickersteth2021pricing, brace2024sofr, fang2024drives, fontana2024term, gellert2021short, lyashenko2019looking, schlogl2024term, skov2021dynamic} and their references. Statistical analysis in \cite{fang2024drives} finds that jumps in SOFR are driven by macroeconomic factors, but their stochastic term structure model does not capture the effect. Closer to the present work are the discrete-time models of \cite{ang2003no, lemke2008affine} which describe the joint dynamics of the term structure and relevant macroeconomic variables. 
In incomplete markets, the martingale measure (or the pricing kernel) is not unique and the chosen one is not necessarily the one that best captures statistical features of the term structure. Like our model, those of \cite{ang2003no} and \cite{lemke2008affine} describe the influence of inflation and GDP growth on the central bank policy rate but, contrary to the main role of central banking, the model of \cite{ang2003no} does not describe the influence of interest rates on the macroeconomy. Moreover, the models of both \cite{ang2003no} and \cite{lemke2008affine} involve latent variables, which makes the calibration of the models complicated. Our model, on the other hand is very easy to calibrate to both historical data, current term structure and user's views. Neither \cite{ang2003no} nor \cite{lemke2008affine} describe the possible jumps on the FOMC meeting dates.

The rest of this paper is organized as follows. Section~\ref{sec:str} constructs a SOFR term structure from futures quotes. Section~\ref{sec:par} proposes parameterizations of the term structure that can be calibrated to available futures quotes by simple quadratic minimization. Section~\ref{sec:tsm} introduces the joint model for the term structure and the macroeconomic variables. Based on the model built in Section~\ref{sec:tsm}, Section~\ref{sec:simstudy} presents a simulation study with indifference pricing of some SOFR derivatives. Section~\ref{sec:arb} takes a brief look at no-arbitrage conditions in discrete-time term structure models.

\section{The SOFR term structure}
\label{sec:str}

A {\em term structure} of interest rates specifies, at a given time and state, the \enquote{prices} of zero-coupon bonds for all maturities. If the zero-coupon bonds were tradeable, their prices should be consistent with observed quotes on interest rate derivatives. In practice, however, only finite numbers of zero-coupon bonds and derivatives are traded, and they all come with bid-ask spreads. 
Nevertheless, approximate term structure models can be useful in approximating prices of bonds and interest rate derivatives. This section reviews the specifications of the exchange-traded SOFR futures contracts and defines the overnight forward curve to be modeled in later sections.

\subsection{SOFR averages and futures}\label{sec:af}

While SOFR only reflects the financing rates observed for the most recent night, the late LIBOR rates were calculated for several maturities. These \enquote{forward-looking} questionnaire-based rates have been replaced by rates on exchange-traded SOFR futures contracts. Currently, CME Group lists 39 three-month futures contracts whose underlyings are geometric averages of the daily SOFR over three-month reference periods. In addition to the three-month futures, there are 13 one-month futures whose underlyings are arithmetic averages of SOFR over one-month reference periods.

The {\em geometric SOFR Average} for a reference period $[t_0,t_1]$ is defined by
\begin{equation}\label{eq:SOFR_Average}
R(t_0,t_1) := \frac{1}{(t_1-t_0)\delta} \left[ \prod_{t = t_0}^{t_1-1} \left(1 + r_t \delta \right) - 1\right],
\end{equation}
where $\delta:=1/360$, and $r_t$ is the SOFR rate applicable from calendar day $t$ to $t+1$ and published on $t+1$. The average $R(t_0,t_1)$ is observed at time $t_1$ when the last SOFR rate $r_{t_1-1}$ is published. SOFR is published only on business days. For bank holidays and weekends, the most recent business day rate is used in the calculation of the average.

For a three-month SOFR futures contract, the reference period $[t_0,t_1]$ starts on the third Wednesday of the third calendar month before the delivery month and ends on the third Wednesday of the delivery month.
According to the CME Rulebook \cite[Section~46003.A]{CMESOFRFut3m},
the payout of a long position is delivered at time $t_1$ and given by
\begin{equation}\label{eq:3mfutpo}
C = [F(t_0,t_1)-R(t_0,t_1)]\times 100\times \$ 2500
\end{equation}
where $F(t_0,t_1)$ is the {\em futures rate} agreed between the counterparties. The exchange-traded three-month futures rates come with bid-ask spreads. We will denote the bid- and ask-rates by $F^b(t_0,t_1)$ and $F^a(t_0,t_1)$, respectively. If we want to emphasize that the quotes were observed at time $t$, we will add a subscript and write $F_t^b(t_0,t_1)$ and $F_t^a(t_0,t_1)$.

\begin{remark}\label{rem:imm}
At CME, SOFR futures are quoted in terms of the \enquote{futures price} $f_t(t_0,t_1)$ which is related to the futures rate above by $F_t(t_0,t_1)=1-f_t(t_0,t_1)/100$; see \cite{imm}.
\end{remark}

The {\em arithmetic SOFR Average} for a reference period $[t_0,t_1]$ is defined by
\[
\hat R(t_0, t_1) := \frac{1}{(t_1-t_0)\delta} \sum_{t=t_0}^{t_1-1} r_t \delta.
\]
In a one-month SOFR futures contract, the reference period $[t_0,t_1]$ starts on the first calendar day and ends on the last day of the month. 
The payout of a long position in a one-month futures contract is delivered at time $t_1$ and given by
\[
C^1 = [F(t_0,t_1) - \hat R(t_0, t_1)] \times 100 \times \$4167
\]
where $F(t_0,t_1)$ is the {\em futures rate} agreed between the counterparties.

\begin{remark}\label{rem:sofravgtaylor}
Using the first order Taylor approximation $\ln x\approx x-1$ twice gives 
\begin{align*}
R(t_0,t_1) &\approx \frac{1}{(t_1-t_0)\delta} \ln\prod_{t=t_0}^{t_1-1} \left(1 + 
r_t \delta \right)\\
&=\frac{1}{(t_1-t_0)\delta} \sum_{t=t_0}^{t_1-1} \ln\left(1 + r_t \delta \right)\approx \hat R(t_0,t_1)
\end{align*}
so one can view the arithmetic average $\hat R(t_0,t_1)$ as an approximation of the geometric average $R(t_0,t_1)$.
\end{remark}

\subsection{The overnight forward curve}\label{sec:forward}

Assuming that, each day, a trader can invest at the overnight rate published the following morning,
the geometric SOFR average $R(t_0,t_1)$ can be thought of as the interest earned from investing $1/[(t_1-t_0)\delta]$ dollars in the overnight repo market over the period $(t_0,t_1)$. Assuming further that zero coupon bonds are available for maturities $t_0$ and $t_1$, the payout of $R(t_0,t_1)$ could then be replicated by a long position of $1/[(t_1-t_0)\delta]$  units in a zero-coupon bond (ZCB) maturing at $t_0$ and a short position of $1/[(t_1-t_0)\delta]$ units in a ZCB maturing at $t_1$. Thus, if the ZCBs with maturities $t_0$ and $t_1$ could be traded at prices $P(t_0)$ and $P(t_1)$, respectively, then the payment $R(t_0,t_1)$ of the floating leg of the futures contract at time $t_1$ could be replicated at cost
\[
\frac{P(t_0)-P(t_1)}{(t_1-t_0)\delta}.
\]
On the other hand, the payout of the fixed leg of a futures contract could be replicated by $F(t_0,t_1)$ units of the zero-coupon bond maturing at $t_1$. If the futures contract and the zero-coupon bonds could be traded without bid-ask spreads, then the unique arbitrage-free futures rate would satisfy
\[
\frac{P(t_0)-P(t_1)}{(t_1-t_0)\delta} = F(t_0,t_1)P(t_1)
\]
or, equivalently,
\begin{equation}\label{pf}
\frac{P(t_0)}{P(t_1)} = 1+F(t_0,t_1)(t_1-t_0)\delta.
\end{equation}

When $t_0=0$, equation \eqref{pf} becomes
\begin{equation}\label{zcbf}
\frac{1}{P(t_1)} = 1+F(0,t_1)t_1\delta
\end{equation}
so $F(0,t_1)$ is just the simply compounded zero-rate for maturity $t_1$. In particular, $F(0,1)$ is the SOFR rate observed on the morning of day $1$.
Here and in what follows, we make the assumption that, when making a decision to lend or borrow in the overnight market at time $t=0$, the investor observes the price $P(1)$ and thus, the SOFR $F(0,1)$ that applies to the night ahead. This may seem to contradict the SOFR publication schedule but it can be justified by the fact that investors participating in the overnight lending market have additional information about the available rates. The overnight lending market is, of course, illiquid in the sense that different rates apply to different lending/borrowing decisions but we take the SOFR rate published in the morning as a proxy for the mid-rate available the day before.

If we are given ZCB prices for maturities $t$ and $t+1$, we define the corresponding {\em overnight forward rate} $F(t)$ by
\[
\frac{P(t)}{P(t+1)} = e^{F(t)\delta}.
\]
If we have perfectly liquid zero-coupon bonds maturing every day up to day $t$, we can express their prices in terms of the overnight forward rates as
\begin{equation}\label{zcb}
P(t) = e^{-\sum_{s=0}^{t-1}F(s)\delta}.
\end{equation}
If the ZCB market was perfectly liquid without bid-ask spreads, this would give a one-to-one correspondence between the zero curve $t\mapsto P(t)$ and the forward curve $t\mapsto F(t)$.

Substituting \eqref{zcb} in \eqref{pf}, we find the following relation between SOFR futures rates and the overnight forward rates:
\begin{equation}\label{eq:FF}
1 + F(t_0,t_1)(t_1-t_0)\delta = \frac{P(t_0)}{P(t_1)} = \exp\left(\sum_{t=t_0}^{t_1-1}F(t)\delta\right).
\end{equation}
In particular, the current overnight forward rate $F(0)$ and the SOFR rate $F(0,1)$ published in the morning of day $1$ are related by
\begin{equation}\label{sofrf}
1+F(0,1)\delta = e^{F(0)\delta},
\end{equation}
which is just a restatement of \eqref{zcbf}. An overnight forward curve $t\mapsto F(t)$ is consistent with the rates $F(t_0,t_1)$ of perfectly liquid futures contracts if it satisfies the {\em linear} equation
\begin{equation}\label{eq:consistency}
\sum_{t=t_0}^{t_1-1}F(t)\delta = \ln(1 + F(t_0,t_1)(t_1-t_0)\delta)
\end{equation}
for every reference period $(t_0,t_1)$ for which futures quotes are available.

\section{Parameterization of the forward curve}\label{sec:par}

In practice, the term structure of interest rates is relevant only up to a certain cut-off date $\bar T$. For example, if one wishes to estimate the forward curve five years into the future, we would have $\bar T=5 \times 365$. The number of futures and swap quotes in practice is not enough to uniquely fix the whole forward curve $(F(t))_{t=0}^{\bar T}$ through condition \eqref{eq:consistency}. 
Following \cite[Section~9.3]{fil9}, we will approximate the overnight forward curve by an {\em affine model} of the form
\begin{equation}\label{parf}
F(t) = \sum_{k=1}^K\xi^k\phi^k(t)\quad t=1,2,\ldots
\end{equation}
where $\phi^k:\naturals\to\reals$ are given \enquote{basis functions} on the timeline and $\xi^k$ are parameters to be determined. The regularity of the basis functions (continuity, differentiability,~\ldots) determines the regularity of the resulting forward curve. Substituting \eqref{parf} into \eqref{eq:consistency} gives 
\begin{equation}\label{eq:consistency2}
\sum_{t=t_0}^{t_1-1}\sum_{k=1}^K\xi^k\phi^k(t)\delta = \ln(1 + F(t_0,t_1)(t_1-t_0)\delta).
\end{equation}
Any forward curve of the form \eqref{parf} has to satisfy this equation to be consistent with the futures quote $F(t_0,t_1)$. One should recall, however, that equations \eqref{eq:consistency} and \eqref{eq:consistency2} assume that neither the forward quotes nor the ZCB prices have bid-ask spreads.

If we have a collection $(F(t_0^j,t_1^j))_{j=1}^J$ of futures quotes, we will get a system of $J$ equations for the parameter vector $\xi=(\xi^k)_{k=1}^K$. If there are more quotes than parameters i.e., if $J>K$, the system is overdetermined in general. We will thus seek a parametric forward curve \eqref{parf} that satisfies the consistency condition \eqref{eq:consistency} for all the futures contracts as well as possible. If we measure the quality of the fit by the sum of squared errors, we arrive at the least squares problem
\[
\minimize\quad \sum_{j=1}^J \left|\sum_{t=t_0^j}^{t_1^j-1}\sum_{k=1}^K\xi^k\phi^k(t)\delta - \ln(1 + F(t_0^j,t_1^j)(t^j_1-t^j_0)\delta)\right|^2\quad\ovr\xi\in\reals^K.
\]
Note that equation \eqref{eq:consistency2} is linear in the parameters $\xi=(\xi^k)_{k=1}^K$ so the above least squares problem is quadratic in the parameters $\xi$. Besides futures quotes, one could use quotes on zero-coupon bonds in the calibration. Indeed, zero-coupon bonds imply futures rates $F(t_0,t_1)$ through \eqref{pf} so the corresponding consistency conditions would also be linear. Swap contracts, however, would result in nonlinear consistency conditions for the forward rates so their inclusion in the least squares estimation would result in a non-convex optimisation problem. Such problems are NP-hard, in general; see e.g.~\cite{nes18}. 

The calibration problem can be written in the vector format
\begin{equation}\label{lsmid}
\minimize\quad \|A\xi-b\|^2\quad\ovr\xi\in\reals^K,
\end{equation}
where $\|\cdot\|$ denotes the Euclidean norm,  $A\in\reals^{J\times K}$ is the matrix with entries
\[
A_{j,k} = \delta\sum_{t=t_0^j}^{t_1^j-1}\phi^k(t)
\]
and $b\in\reals^J$ is a vector with components
\[
b_j = \ln(1 + F(t_0^j,t_1^j)(t^j_1-t^j_0)\delta).
\]
If $A$ has linearly independent columns, the problem has a unique solution given by
\[
\xi = (A^TA)^{-1}A^Tb.
\]
If the least squares solution fails to be unique, it is natural to choose the one with the least Euclidean norm. The minimum norm solution always exists and it can be expressed as
\[
\xi=A^+b,
\]
where $A^+$ is the Moore-Penrose inverse of $A$.

In reality, the forward quotes always come with bid-ask spreads so the consistency equation \eqref{eq:consistency} should be replaced by
\begin{equation}\label{eq:consistency3}
b^b \le
\sum_{t=t_0}^{t_1-1}F(t)\delta \le b^a,
\end{equation}
where
\[
b^b:=\ln(1 + F^b(t_0,t_1)(t_1-t_0)\delta)\quad\text{and}\quad b^a:=\ln(1 + F^a(t_0,t_1)(t_1-t_0)\delta)
\]
and $F^b(t_0,t_1)$ and $F^a(t_0,t_1)$ are the bid- and ask-futures rates, respectively. Accordingly, condition \eqref{eq:consistency2} on the parametric forward curve should be replaced by
\begin{equation}\label{eq:consistency4}
b^b\le
\sum_{t=t_0}^{t_1-1}\sum_{k=1}^K\xi^k\phi^k(t)\delta \le b^a,
\end{equation}
In this case, the least squares problem can be written as
\begin{equation}\label{ls}
\begin{aligned}
& \minimize\quad & &\|e\|^2\quad\ovr \xi\in\reals^K,\ e\in\reals^J,\\
& \st\quad & & b^b\le A\xi+e\le b^a,
\end{aligned}
\end{equation}
where $b^b$ and $b^a$ are now vectors whose components are defined according to the available futures quotes. The components of the vector $e$ quantify the violation of the lower and upper bounds in \eqref{eq:consistency4} for each of the quoted forward contracts. 

Problem \eqref{ls} is that of convex quadratic optimization so it can be solved quickly by off-the-shelf optimization solvers. In this paper, we employ the interior point solver of Mosek~\cite{aps2019mosek}. The problem was formulated and communicated to Mosek using Python~\cite{van1995python} and CVXPY~\cite{diamond2016cvxpy}. The solution is obtained in less than 0.4 seconds using a Dell Latitude running Linux with 16 GB of RAM and a 4-core Intel i7 CPU at 3.00 GHz; see Example~\ref{ex:contpar}. Similarly, solving 1300 instances of problem~\eqref{lsmid} takes less than 8 minutes; see Section~\ref{sec:datatr}.

\begin{example}[Piecewise constant parametrization]\label{ex:cmepar}
Let $(T_k)_{k=1}^K$ be a sequence of increasing points in time and consider the parameterization \eqref{parf} with the piecewise constant basis functions
\[
\phi^k(t)=
\begin{cases}
    1 & t\in[T_{k-1},T_k],\\
    0 & t\notin[T_{k-1},T_k].
\end{cases}
\]
Such basis functions result in piecewise constant forward curves which are used e.g.\ by CME \cite{cme}.
Like CME, we use one-month and three-month futures quotes and SOFR to estimate the forward curve. However, instead of mid-quotes, we use the bid and ask quotes, and solve the least squares problem \eqref{ls} with tenors corresponding to FOMC meeting dates
\[
(T_k)_{k=1}^{K} = (0, 1.2m, 2.5m, 3.9m, 5.5m, 7m, 8.5m, 9.7m, 11.5m, 1y).
\] 
The quotes are illustrated in Figure~\ref{fig:termCMEstep} (without distinguishing between one- and three-month futures quotes).
\end{example}

\begin{figure}[!ht]
	\centering
	{ \includegraphics[trim = 0mm 0mm 0mm         0mm, clip, width=0.6\textwidth] 
        {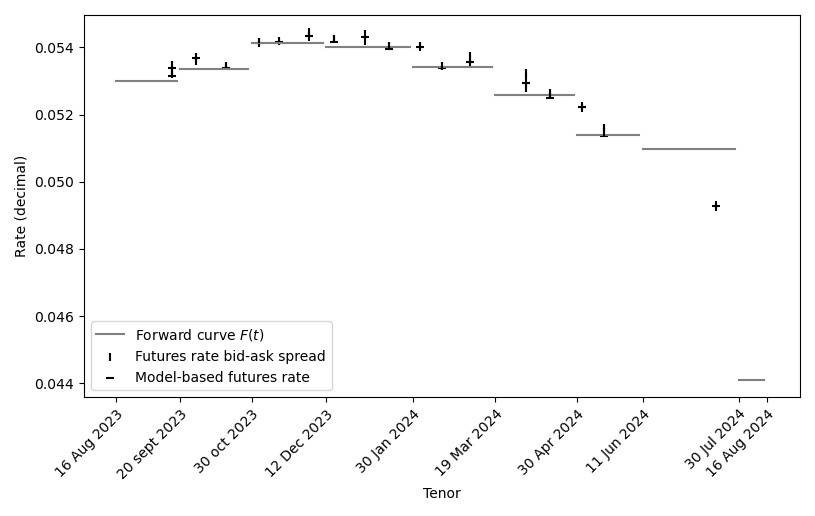}
	}
	\centering
	\caption
	{Estimated forward curve 12 months ahead on 16 August 2023 with a cut-off date of 1 year using piecewise constant basis functions from Example \ref{ex:cmepar}. The tenors correspond to the FOMC meeting dates after 16 August 2023. The model-based forward rates are computed from the forward curve 
    $F(t)$ using \eqref{eq:consistency}. 
    The futures rate quotes are computed using Remark~\ref{rem:imm} and are plotted in the middle of the reference quarter ({\em data source: Bloomberg Finance L.P.})
	}
	\label{fig:termCMEstep}
\end{figure}

\begin{example}[Continuous parameterization]\label{ex:contpar}
We will now use the parameterization \eqref{parf} with the continuous basis functions $\phi^k$ which are linear between consecutive tenors $T_k$ and satisfy
\begin{equation}\label{eq:piecewise}
\phi^k(t)=
\begin{cases}
    1 & t=T_k,\\
    0 & t\notin[T_{k-1},T_{k+1}].
\end{cases}
\end{equation}
This results in continuous piecewise-linear forward curves. Indeed, the corresponding forward curve given by \eqref{parf}, has $F(T_k)=\xi^k$ for all $k$ and the curve interpolates linearly between consecutive tenors. Continuity is often seen as a desirable property in term structure modelling; see e.g.\ \cite[Section~3.3]{fil9} and the references there. Having an economic interpretation for the parameters $\xi^k$, facilitates their statistical modelling; see Section~\ref{sec:tsm}.

We use the SOFR and available one- and three-month SOFR futures quotes with maturity up to 5 years and estimate a piecewise linear forward curve with tenors
\[
(T_k)_{k=1}^K = (0,1m,3m,6m,1y,2y,3y, 4y,5y).
\]
In particular, the first component of $\xi$ is the one-day maturity zero-rate $F(0)$ which is essentially the current SOFR; see \eqref{sofrf}. We fit the above parameterization to the observed futures quotes as well as to the value of SOFR observed the morning after. We do this by solving problem~\eqref{ls} with the additional constraint that $\xi^1=\ln(1+F(0,1)\delta)/\delta$; see equation \eqref{sofrf}. This results in the piecewise linear curve in Figure~\ref{fig:fwdcurve}. The figure also plots the employed futures bid- and ask-quotes obtained from Bloomberg Finance L.P. Because the SOFR rate is published on the next business day, we shift the horizontal position of the quotes 1 day backwards to match with the futures prices that are published on the trading day. 
\end{example}

\begin{figure}[!ht]
	\centering
	{\includegraphics[trim = 0mm 0mm 0mm 0mm, clip, width=0.6\textwidth] 
            {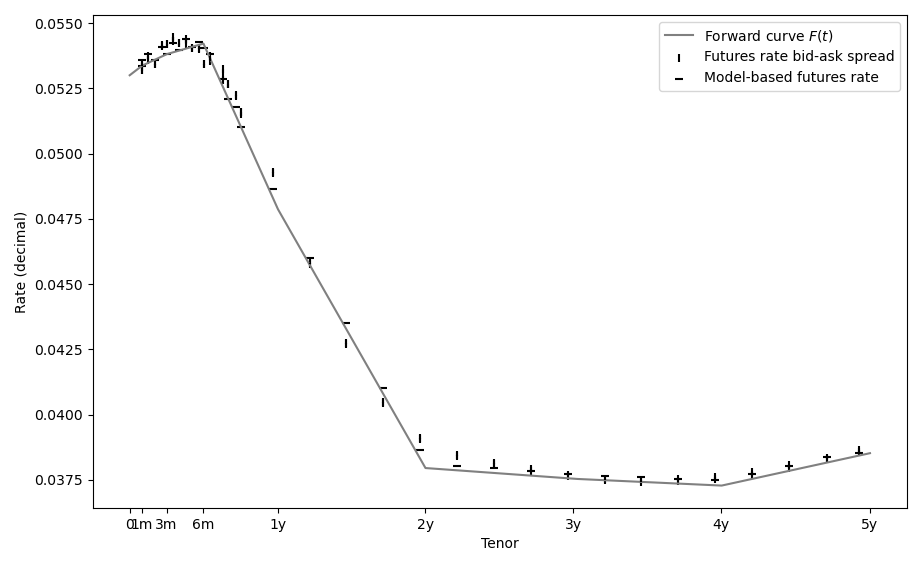}
	}
	\centering
	\caption
	{Estimated forward curve five years ahead on 16 August 2023 using continuous basis functions from Example \ref{ex:contpar}. The model-based forward rates are computed from the forward curve $F(t)$ using \eqref{eq:consistency}.  
    The rates over the first 12 months differ from those in Figure~\ref{fig:termCMEstep}
    because of the different tenors and basis functions. The futures rate quotes are computed using Remark~\ref{rem:imm} and are plotted in the middle of the reference quarter ({\em data source: Bloomberg Finance L.P.})
	}
	\label{fig:fwdcurve}
\end{figure}

\section{Statistical modelling of the forward curve}\label{sec:tsm}

Up to now, we have been studying the term structure at a single point in time. We will next focus on the dynamics of the term structure and denote the forward curve at time $t$ by $F_t:=(F_t(s))_{s=0}^T$. In order to describe the evolution of $F_t$, we will, much like in Chapter~9 of \cite{fil9}, first approximate it by the parameterization \eqref{parf} and then model the parameters as a multivariate stochastic process. More precisely, we will assume that $F_t$ has, for every day $t=0,1,2,\ldots$, the parametric representation
\begin{equation}\label{eq:fcpar}
F_t(s) = \sum_{k=1}^K\xi^k_t\phi^k(s-t)\quad s\ge t,
\end{equation}
where $\xi^k:=(\xi^k_t)_{t=0}^\infty$ are real-valued stochastic processes. We will denote by $\xi$ the $\reals^K$-valued stochastic process with components $\xi^k$. This can be seen as a parametric term structure model in discrete time much as e.g.\ in \cite[Section~9.3]{fil9}. 
While term-structure models are usually specified under a fixed risk-neutral measure, we will describe the behavior of the term structure under a \enquote{statistical} measure (also known as the \enquote{$P$-measure} or the \enquote{real-world measure}) which will be more suitable for portfolio optimization and indifference pricing of SOFR derivatives. 
We will follow the general procedure described in~\cite[Sections~3 and 4]{aap21} to develop a stochastic model for $\xi$. This allows us to calibrate the model to both historical data and user's views concerning the future development of the term structure.

\subsection{SOFR and policy rates}
\label{sec:sofrpr}

SOFR is calculated daily as a volume-weighted median rate of overnight repurchase transactions (repo) of depositary institutions and the FED; see \cite{sifmarepo}.
The FED offers a deposit rate, known as overnight Reverse Repurchase Agreement Facility (ON RRP) since September 2014, and a lending rate, known as Standing Overnight Repurchase Agreement Facility (SRF) since July 2021; see Table~\ref{table:fedrates}. 
Most participants in the repo market have access to these Fed facilities (see \cite{RRepoCP, SRepoCP, secrepo}) so they are unlikely to lend to a counterparty at a rate lower than ON RRP.
Similar to the federal funds target range $[L, H]$, SRF and ON RRP are set during Federal Open Market Committee (FOMC) meetings and remain constant between meetings. Historically, ON RRP was set above $L$, except for a brief period between the end of 2019 and the beginning of 2020; see Figure~\ref{fig:effr}. 
When enough liquidity exists in the FED reverse repo facility, one would expect SOFR to stay above $L$ as well. The rare cases where SOFR falls below $L$ may be explained by the demand for ON RRP exceeding what the FED is offering, or some participants, which do not have access to FED facilities, accepting deposit rates lower than ON RRP. The converse applies to the federal funds target range upper limit $H$. Table~\ref{table:fedrates} illustrates the situation on 4 January 2024.

When the range $[L, H]$ is changed, SOFR will follow. This explains the jumps in SOFR; see Figure~\ref{fig:effr}. 
In addition, SOFR was prone to spikes at regulatory reporting dates until July 2021 when SRF was established.

\begin{table}[ht]
	\centering
	\caption
		{
		\label{table:fedrates}	
		A glimpse of FED monetary policy measures as of 4 January 2024. 
        SOFR was $5.32\%$
		\emph{(source: FED \cite{FederalReserveWebsite}.)}
		}
                    \vskip 3mm
				\resizebox{11cm}{!}
				{
                        \rowcolors{2}{}{mygray}
					\begin{tabular}{l l l}

						%---------------------
						\toprule
						\textbf{Rate}
						& \textbf{Value} 
						& \textbf{Funds available}
						\\

                            %---------------------
						\midrule
                            Federal funds target range $[L, H]$
                            & $[5.25\%, 5.50\%]$
                            & 
                            \\
						%---------------------
						% \midrule
                            SRF minimum bid rate
                            & $5.50\%$
                            & \$500 billion of aggregate operation per counterparty per day
                            \\

                            %---------------------
						% \midrule
                            ON RRP offering rate
                            & $5.30\%$
                            &\$160 billion of aggregate operation per counterparty per day
                            \\

						%---------------------
						\bottomrule
						
					\end{tabular}
				}
                    \vspace{2mm}

				\vskip 3mm
\end{table}

\begin{figure}[!ht] 
	\centering
	{\includegraphics[trim = 0mm 0mm 0mm 0mm, clip, width=0.6\textwidth] 
            {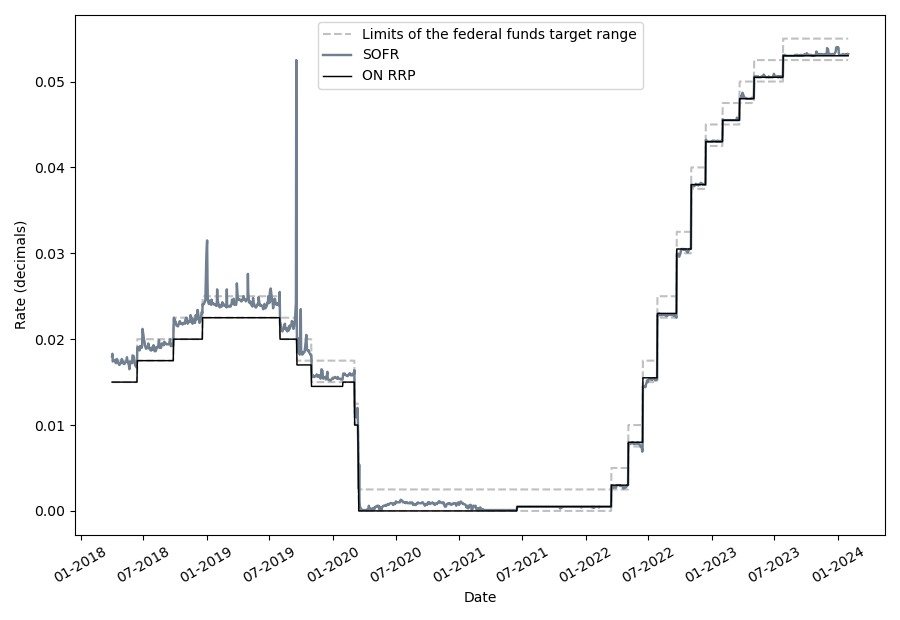}
	}
	\centering
	\caption
	{Daily values of SOFR, ON RRP, and the limits of federal funds target range $[L, H]$ ({\em source: FRED} \cite{SOFR, DFEDTARU, DFEDTARL, RRPONTSYAWARD}).
	}
	\label{fig:effr}
\end{figure}

\subsection{The data and the transformations} \label{sec:datatr}

We will describe the SOFR term structure using the piecewise linear basis functions $\phi^k$ from Example~\ref{ex:contpar} and find historical daily values for the parameters $\xi$ by taking daily SOFR and closing prices for three-month futures and solving problem \eqref{lsmid} for each day separately.  
We will use the tenors from Example~\ref{ex:contpar} but other choices would be equally easy to use. 
Daily closing prices spanning these tenors are available from Bloomberg Finance L.P.\ from 1 April 2020 onwards.
Note that the large spikes in SOFR disappeared after July 2021 when SRF was established; see Figure~\ref{fig:effr}. 
Recall that our choice of the basis functions $\phi^k$ implies that $\xi^k$ is the value of the forward curve at maturity $T_k$. Figure~\ref{fig:xi} plots the obtained time series of daily values of $\xi^k$ starting from 1 April 2021. 

The components of $\xi$ appear nonstationary over the examined period. This is confirmed by the results of the augmented Dickey–Fuller test reported in Table~\ref{table:ADF_xi}. As explained in Section~\ref{sec:sofrpr} above, SOFR tends to follow the federal funds target range $[L, H]$, which has increased over the studied period. The first coefficient $\xi^1$ coincides with the value $F_t(0)$ of the forward curve at maturity of one day, which is related to SOFR through formula \eqref{sofrf} which is very close to an identity mapping for any reasonable value of the SOFR rate. As a result, $\xi^1$ tends to have only small variations between consecutive FOMC meeting dates but it jumps whenever the target is changed. To capture this feature of the market, we will model the variable
\begin{equation}\label{eq:x_1}
x_t^1 := \ln (\xi^1_t + c^1 - L_t),
\end{equation}
as a stationary process. Here $c^1$ is a small positive constant which accounts for the possibility of SOFR falling slightly below the target range.

The remaining parameters $\xi^k$ for $k=2,\ldots,K$  are also points on the overnight forward curve, so they tend to move in tandem with the SOFR. For nearer maturities the link tends to be stronger. This suggests modelling the parameters $\xi^k$, $k=1,\ldots,K$ through the differences
\begin{equation}\label{eq:x_k}
x_t^k := \ln(\xi^k_t+c^k)-\ln(\xi^{k-1}_t+c^{k-1}),
\end{equation}
where again, $c^k$ are small positive constants that allow for slightly negative values of the forward rates $F_t(T_k)=\xi^k$. The values of all the constants $c^k$ are chosen to make all the series $x^k$ as stationary and symmetric as possible. The values are given in Table~\ref{table:xishifts}. Since  $\xi^k_t=F_t(T_k)$, the variables $x^k_t$ for $k=2,\ldots,K$ can be thought of as the \enquote{term spreads} over the \enquote{tenors} $[t+T_{k-1},t+T_k]$.

\begin{table}[ht]
	\centering
	\caption
	{
		\label{table:xishifts}
		The shifts $c^k$ of $\xi^k$. The values have been multiplied by $100$ for readability.
	}
	\resizebox{8cm}{!}
	{
            \begin{tabular}{SSSSSSSSSS}
			\toprule
			{$c^1$}
			& {$c^2$}
			& {$c^3$}
			& {$c^4$}
			& {$c^5$} 
			& {$c^6$}
			& {$c^7$}
			& {$c^8$} 
			& {$c^9$}  
			\\	
                \midrule
			0.81
                & 0.965
                & 0.954
                & 0.9
                & 0.67
                & 0.44
                & 0.15
                & 0.042
                & - 0.02
			\\
			\bottomrule
		\end{tabular}
	}
\end{table}

Figures~\ref{fig:x_1} and~\ref{fig:x_k} plot the historical daily values of $x=(x^1,\ldots,x^K)$.
Compared to $\xi^k$, the time series $x$ appears more stationary. This is confirmed by the results of the augmented Dickey-Fuller test in Table~\ref{table:ADF_x} where the values of the test statistic are further from zero as compared to the values in Table~\ref{table:ADF_xi}. Most components of $x$ still fail the stationarity test at the usual 5\%-confidence level but given the relatively short observation period, one should not draw strong conclusions from that. Indeed, it is economically reasonable to assume that all components of $x$ are stationary in the long run. Stationarity of the first component $x^1$, in particular, means that SOFR tends to follow the federal funds target range, an essential feature of the market. Stationarity of the remaining components of $x$ means that the forward rates tend to follow the SOFR.

\begin{table}[ht]
	\centering
        \caption{\label{table:ADF_xi}
	Augmented Dickey–Fuller test for the term structure risk factors (p-values are in parenthesis).
	}
	\resizebox{11cm}{!}
	{
            \begin{tabular}{lSSSSSSSSSS}
			%---------------------
			\toprule
			& {$\xi^1$}
			& {$\xi^2$}
			& {$\xi^3$}
			& {$\xi^4$}
			& {$\xi^5$} 
			& {$\xi^6$}
			& {$\xi^7$}
			& {$\xi^8$} 
			& {$\xi^9$}  
			\\
			
			%---------------------
                \midrule
			{ADF t-statistic}
                & 0.673
                & 0.612
                & 0.100
                & -0.404
                & -0.924
                & -1.070
                & -1.074
                & -1.168
                & -1.207
			\\

                &(0.989)
                &(0.988)
                &(0.966)
                &(0.909)
                &(0.780)
                &(0.727)
                &(0.725)
                &(0.687)
                &(0.671)
			\\

			%---------------------
			\bottomrule
		\end{tabular}
	}
	
\end{table}

\begin{table}[ht]
	\centering
	\caption
	{
		\label{table:ADF_x}
		Augmented Dickey–Fuller test for the transformed term structure risk factors (p-values are in parenthesis).
	}
	
	\resizebox{11cm}{!}
	{
            \begin{tabular}{lSSSSSSSSSS}
			%---------------------
			\toprule
			& {$x^1$}
			& {$x^2$}
			& {$x^3$}
			& {$x^4$}
			& {$x^5$} 
			& {$x^6$}
			& {$x^7$}
			& {$x^8$} 
			& {$x^9$}  
			\\
			
			%---------------------
                \midrule
			{ADF t-statistic}
                & -2.599
                & -2.938
                & -1.632
                & -2.343
                & -1.783
                & -1.221
                & -1.745
                & -1.857
                & -2.903
			\\

                &(0.093)
                &(0.041)
                &(0.466)
                &(0.159)
                &(0.389)
                &(0.664)
                &(0.408)
                &(0.353)
                &(0.045)
			\\

			%---------------------
			\bottomrule
		\end{tabular}
	}
	
\end{table}

\begin{table}[ht]
	\centering
	\caption
	{
		\label{table:ADF_y}
		Augmented Dickey–Fuller test for $y$ (p-values are in parenthesis).
	}
	
	\resizebox{5cm}{!}
	{
            \begin{tabular}{lSSSSSSSSSS}
			%---------------------
			\toprule
			& {$\ln (L + c^L)$}
			& {$I$}
			& {$G$} 
			\\
			
			%---------------------
                \midrule
			{ADF t-statistic}
                & -2.288 
                & -2.901
                & -5.351
			\\

                &(0.176)
                &(0.045)
                &(0.000)
			\\

			%---------------------
			\bottomrule
		\end{tabular}
	}
	
\end{table}

\begin{figure}[!ht]
	\centering
	{
		\includegraphics[trim = 0mm 0mm 0mm 0mm, clip, width=0.6\textwidth] 
		{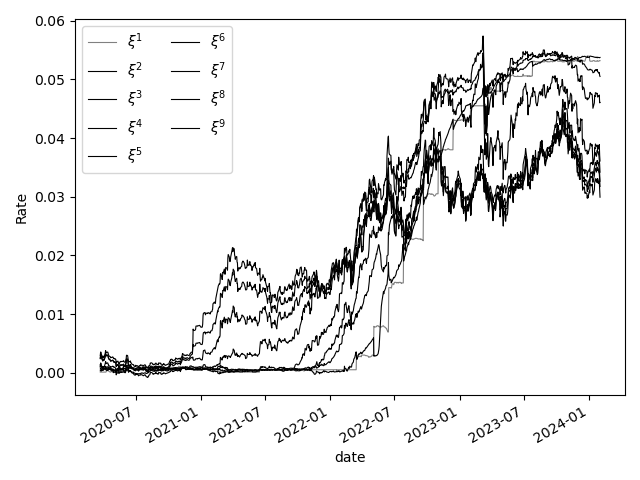}
	}
	\centering
	\caption
	{Historical daily values of the term structure risk factors $\xi$, the components of which are values of the overnight forward curve at $(0,1m,3m,6m,1y,2y,3y, 4y,5y)$; see Example~\ref{ex:contpar}. ({\em data source: Bloomberg Finance L.P.})
	}
	\label{fig:xi}
\end{figure}

\begin{figure}[!ht]
	\centering        
	{\includegraphics[trim = 0mm 0mm 0mm 0mm, clip, width=0.6\textwidth] 
		{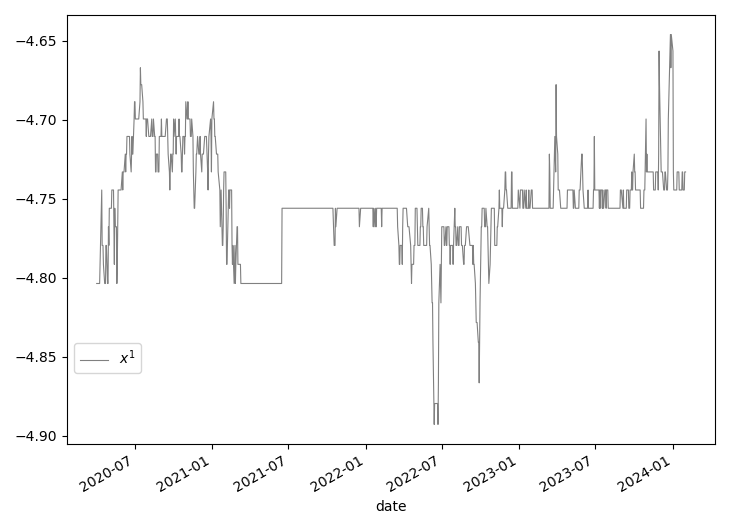}
	}
	\centering
	\caption
	{Historical values of $x^1:=\ln(\xi^1-L)$. Here, $\xi^1$ is the spot rate $F(0)$ and $L$ is the lower limit of the FED target range. The SOFR moves with multiples of 1 basis point, which causes the discrete jumps in the time series.
	}
	\label{fig:x_1}
\end{figure}

\begin{figure}[!ht]
	\centering      
	{\includegraphics[trim = 0mm 0mm 0mm 0mm, clip, width=0.6\textwidth] 
		{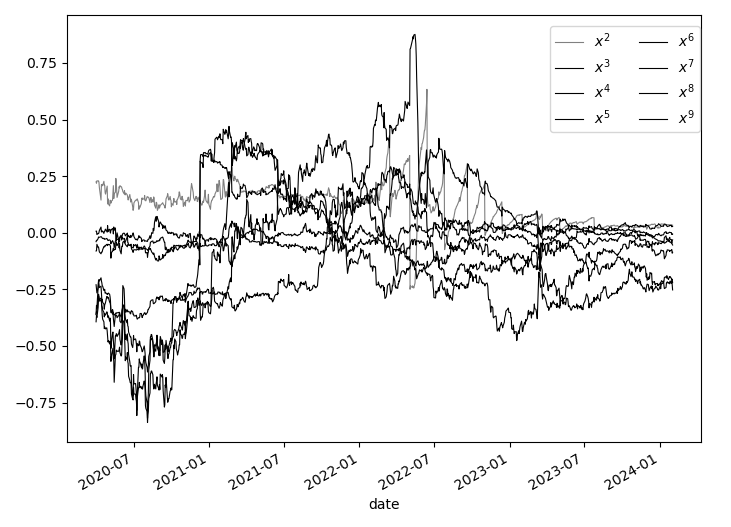}
        }
	\centering
	\caption
	{Historical values of $x^k$, $k=2, \cdots, 9$.
	}
	\label{fig:x_k}
\end{figure}

A complete description of the forward curve dynamics requires a model for the lower limit $L$ of the federal funds target range. There are several macroeconomic factors that may affect the FOMC rate decisions, but the most significant ones are the inflation $I$ and the real GDP growth rate $G$; see e.g.~\cite{mishkin2007economics}. This is in line also with the Taylor rule \cite{taylor1993discretion}, which models the central bank policy rate as a simple parametric function of $G$ and $I$.

We will model the values of $L$, $I$, and $G$ at FOMC meeting dates. Instead of modelling the values of $L$ directly, we will model the variable $y^1:=\ln(L+c^L)$ as a stationary stochastic process. Here, $c^L$ is a small positive constant which allows for the possibility of strictly negative target rates in the future. Such scenarios seem plausible given what happened in the Eurozone after the 2008 financial crisis. We choose the value $c^L = 0.5\%$ which makes $y^1$ fairly stationary and symmetrically distributed. We will develop a stationary stochastic model for the three-dimensional process
\begin{equation}\label{eq:macrotransform}
y_t = [\ln(L_t+c^L),I_t,G_t].
\end{equation}
Figure~\ref{fig:macrots} plots the monthly values of $y$ while Table~\ref{table:ADF_y} reports the results of the augmented Dickey-Fuller test for the components of $y$. Except for $y^1:=\ln (L+c^L)$, the time series appears stationarity at $5\%$ confidence level. The slight nonstationarity of $y^1$ reflects the changes in monetary policy since the beginning of the time series in 1972. Economically, it is reasonable to assume that the decline of interest rates will not continue indefinitely so that stationarity of $\ln (L + c^L)$ in the future seems like a reasonable assumption.

\begin{figure}[!ht]
	\centering
	{\includegraphics[trim = 0mm 0mm 0mm 0mm, clip, width=0.6\textwidth] 
		{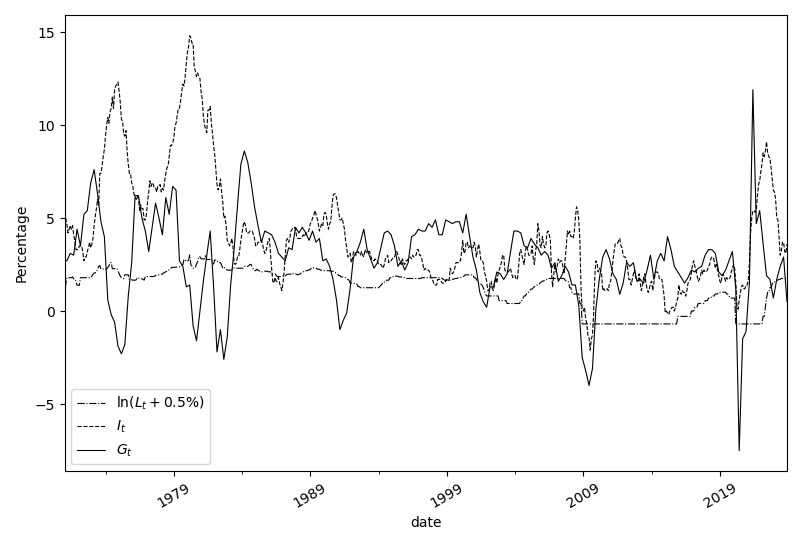}
	}
	\centering
	\caption
	{
		Historical values of the macroeconomic risk factors $\ln (L_t + 0.5 \%)$, $I_t$, and $G_t$ ({\em sources: FRED} \cite{DFEDTARL} and Bloomberg Finance L.P.)
	}
	\label{fig:macrots}
\end{figure}

\subsection{The macroeconomic model}\label{sec:macro}

We model the vector $y := [\ln(L+c^L),I,G]$ by the linear stochastic difference equation (vector autoregressive model)
\begin{equation}\label{eq:var}
\Delta y_t = Ay_{t-1} + a_t + \epsilon_t,
\end{equation}
where $A\in\reals^{3\times 3}$ and $a_t\in\reals^3$ are given parameters and $\epsilon_t$ are iid Gaussian vectors in $\reals^3$ with zero mean and covariance matrix $\Sigma$. The above model specification could be generalized in many ways but, as we will see below, the simple specification above already results in a fairly reasonable description of the macroeconomic variables.

The matrix $A$ as well as the distribution of the innovations $\epsilon_t$ will be estimated from historical data. We use monthly observations of $y$ from January 1972 to January 2024. The monthly frequency can be thought of as an approximation of the FOMC meeting frequency. The data on $G_t$ comes at a quarterly frequency, so we interpolate it to get monthly values. 

We estimate the matrix $A$ using ordinary least squares and set insignificant coefficients to zero. This results in the matrix given in Table~\ref{table:amatrixy}. The negative diagonal elements indicate a mean-reverting behavior of the components of $y$. 
The positive signs of $A_{12}$ and $A_{13}$ suggest that the FED tends to raise the target rate $L$ when the inflation $I$ and/or the real GDP growth rate $G$ increase. This is in line with general monetary policy and, in particular, the Taylor rule \cite{taylor1993discretion}.
The negative sign of $A_{21}$ reflects the fact that higher interest rates tend to slow down inflation and vice versa. The positive sign of $A_{23}$ means that high real GDP growth tends to increase inflation and vice versa. The negative sign of $A_{32}$ just means that higher inflation tends to slow down real GDP growth.

Table~\ref{table:eigen_y} lists the eigenvalues of the matrix $(A+I)$. They all lie within the unit circle in the complex plane. This implies that the process $y$ is stationary. The imaginary values imply oscillations in the time series, which correspond to economic cycles. We use the regression residuals to estimate the distribution of the innovations $\epsilon_t$. For simplicity, we use a multivariate Gaussian distribution. The correlations and the variances are given in Table~\ref{table:corry}. 

The above model is close to the macroeconomic model of \cite[Section 4.3]{ang2003no}, which described the influence of inflation and GDP growth on the interest rate. The model of \cite{ang2003no} did not, however, include the feedback effect from interest rates to the macroeconomy. This effect is an essential feature of the macroeconomy and central banking.

\begin{table}[ht]
	\centering
	\caption
	{
		\label{table:amatrixy}
		Coefficients of the autoregression matrix A for $y$ and their 
		corresponding p-values (in parenthesis). 
        Zero elements correspond to p-values $> 10\%$; they were omitted for clarity.
	}
	
	\resizebox{5cm}{!}
	{                
                \tiny
			\sisetup{table-number-alignment=center}
                \begin{tabular}{lSSS}
				
				%---------------------
				\toprule
				& {$\ln L$}
				& {$I$}
				& {$G$}
		
				\\
				
				%---------------------
				\midrule
				\rowcolor{mygray}
				  & -0.018  
                    & 0.006 
                    & 0.008
                                       
				\\
				
				%---------------------
                    \rowcolor{mygray}
                    {\multirow{-2}{*}{$\ln L $}}
                    & (0.000)
                    & (0.001)  
                    & (0.000)
                    
				\\				
				%---------------------
                    & -0.037 
				&  % {0.000}
				& 0.037
				\\
				
				%---------------------
				{\multirow{-2}{*}{$I$}}
                    & (0.053)
				& (0.000)
				\\
				
				%---------------------
				\rowcolor{mygray}
                    &
				& -0.028
				& -0.032
				\\
				
				%---------------------
				\rowcolor{mygray}
                    {\multirow{-2}{*}{$G$}}
                    &
				& (0.001)
				& (0.001)
				\\
				
				%---------------------
				\bottomrule
			\end{tabular}
		}
		
\end{table}

\begin{table}[ht]
	\centering
	\caption
	{
		\label{table:eigen_y}
		The eigenvalues of the matrix $(A+I)$ in the time series model for the macroeconomic risk factors $y$. The imaginary unit is denoted by $i$.
	}
	
	\resizebox{4.5cm}{!}
	{
                \begin{tabular}{SSS}
			%---------------------
			\toprule
                \multicolumn{3}{c}{Eigenvalues of $(A+I)$}
			\\
			
			%---------------------
                \midrule			
                 0.982
                & {0.984+0.0279i}
                & {0.984-0.0279i}
			\\
   			
			%---------------------
			\bottomrule
		\end{tabular}
	}

\end{table}
	
\begin{table}[ht]
	\centering
	\caption
		{
			\label{table:corry}
	Correlations and variances of the residuals of the time series model for the macroeconomic risk factors $y$
		}
		
		\resizebox{5cm}{!}
		{
                
                \tiny
			\sisetup{table-number-alignment=center}
                \begin{tabular}{lSSS}
				%---------------------
				\toprule
                    
				& {$\ln L$}
				& {$I$}
				& {$G$}
				\\
				
				%---------------------
				\midrule
				\rowcolor{mygray}
				{$\ln L$}
                    & 1.000 
				& 0.013
				& 0.099
    
		          \\
				
				%---------------------
				{$I$}
				& & 1.000
                    & 0.170
				\\
				
				%---------------------
				\rowcolor{mygray}
				{$G$}
				& & & 1.000
                    \\	
                    
				%---------------------
				\midrule
				{Variance}
                    & 0.011
                    & 0.166
                    & 0.271
 
				\\
				%---------------------
				\bottomrule
			\end{tabular}
		}
		
\end{table}

\subsection{The forward curve model}\label{sec:termstruc}

We model the daily values of the 9-dimensional process $x=(x_t)_{t=0}^\infty$ defined in \eqref{eq:x_1} and \eqref{eq:x_k} by another vector autoregressive model of the form 
\begin{equation}\label{eq:varx}
\Delta x_t = Ax_{t-1} + a_t + \epsilon_t,
\end{equation}
where $A\in\reals^{9\times 9}$ and $a_t\in\reals^9$ are given parameters and $\epsilon_t$ are iid Gaussian vectors in $\reals^9$ with zero mean and covariance matrix $\Sigma$. The vector autoregressive model is fitted to daily data constructed from the historical values of the target rate $L$ and the daily closing prices of futures contracts as described in Section~\ref{sec:datatr}.
The parameter matrix and the distribution of the innovations are estimated as described in  Section~\ref{sec:macro} in the context of the macroeconomic variables. However, to simplify the statistical analysis and the modelling, we restrict the autoregressive matrix to be diagonal. One could, of course, allow for general AR-matrices but the interpretation of the off-diagonal coefficients would not be quite straightforward. The diagonal elements, on the other hand, correspond to mean reversion which is a natural feature for the shape of the modelled forward curve. The estimated coefficients are given in Table~\ref{table:amatrixx}. 
All diagonal values are negative which indicates mean reversion and stationarity in the components of $x$. 
Table \ref{table:corrx} reports the correlations and the variances of the residual term.

\begin{table}[ht]
	\centering
	\caption
	{
		\label{table:amatrixx}
		Coefficients of the autoregression matrix A for $x$ and their corresponding p-values (in parenthesis).
	}
	
	\resizebox{9cm}{!}
	{  \sisetup{table-number-alignment=center}
		\begin{tabular}{S[table-format=-0.1]*{10}S}

			%---------------------
			\toprule
			& {$x^1$}
			& {$x^2$}
			& {$x^3$}
			& {$x^4$} 
			& {$x^5$}
			& {$x^6$}
			& {$x^7$}
                & {$x^8$}
                & {$x^9$}
			\\

                %---------------------
			\midrule
			\rowcolor{mygray}
			& -0.059
                & & & & & & & &
			\\
			
			%---------------------
			\rowcolor{mygray}
                {\multirow{-2}{*}{$x^1$}}
			& (0.000)	
                & & & & & & & &
			\\
			
			%---------------------
			
                & & -0.077
			\\
			
			%---------------------
			{\multirow{-2}{*}{$x^2$}}
			& &(0.000)
			\\
			
			%---------------------
			\rowcolor{mygray}
			& & & -0.006
                & & & & & &
			\\
			
			%---------------------
			\rowcolor{mygray}
                {\multirow{-2}{*}{$x^3$}}
			& & & (0.105)
                & & & & & &
			\\		
			
			%---------------------
			
			& & & & -0.011
			\\
			
			%---------------------
			{\multirow{-2}{*}{$x^4$}}
			& & & & (0.019)	
			\\	
			
			%---------------------
			\rowcolor{mygray}
			& & & & & -0.003
                & & & &
			\\
			
			%---------------------
			\rowcolor{mygray}
                {\multirow{-2}{*}{$x^5$}}
			& & & & & (0.218)
                & & & & 
			\\	
			
			%---------------------
			
			& & & & & & -0.004
			\\
			
			%---------------------
			{\multirow{-2}{*}{$x^6$}}
			& & & & & & (0.148)	
			\\	
			
			%---------------------
			\rowcolor{mygray}
			& & & & & & & -0.004
                & &
			\\
			
			%---------------------
			\rowcolor{mygray}
			{\multirow{-2}{*}{$x^7$}}
                & & & & & & & (0.160)
                & &
			\\	
			
			%---------------------
			
			& & & & & & & & -0.005
			\\
			
			%---------------------
			{\multirow{-2}{*}{$x^8$}}
			& & & & & & & & (0.084)
			\\

                %---------------------
			\rowcolor{mygray}
			& & & & & & & & & -0.019
			\\
			
			%---------------------
			\rowcolor{mygray}
                {\multirow{-2}{*}{$x^9$}}
			& & & & & & & & & (0.002)	
			\\
			
			%---------------------
			\bottomrule
		\end{tabular}
	}
	
\end{table}

\begin{table}[ht]
	\centering
	\caption
	{
		\label{table:corrx}
		Correlation coefficients for the residuals of $x$. Variances are in the bottom row. 
	}
	
	\resizebox{9cm}{!}
	{
            \rowcolors{2}{}{mygray}
            \begin{tabular}{lSSSSSSSSSS}
			%---------------------
			\toprule
			& {$x^1$}
			& {$x^2$}
			& {$x^3$}
			& {$x^4$}
			& {$x^5$} 
			& {$x^6$}
			& {$x^7$}
			& {$x^8$} 
			& {$x^9$}  
			\\
			
			%---------------------
			\midrule			
			{$x^1$}
			& 1.000
                & -0.209  
                & 0.144 
                & -0.123  
                & 0.037 
                & -0.032 
                &-0.023  
                &0.002  
                & 0.004
			\\
   
			%---------------------		
			{$x^2$}
			& & 1.000
			& -0.683  
                & 0.277 
                & -0.043  
                & 0.041  
                & 0.034 
                & -0.002 
                & -0.032
			\\
               
			%---------------------
			{$x^3$}
			& & & 1.000
			& -0.303  
                & 0.138 
                & -0.086 
                & -0.052 
                & -0.013
                & -0.035
			\\
			
			%---------------------			
			{$x^4$}
			& & & & 1.000
			& 0.078 
                & -0.047
                & 0.010
                & -0.003 
                & -0.066
			\\		
			%---------------------
			
			{$x^5$}
			& & & & & 1.000
			& 0.164 
                & -0.092
                & 0.000
                & -0.136
			\\
   
			%---------------------
			
			{$x^6$}
			& & & & & & 1.000
			& 0.428  
                & 0.052
                & -0.078
			\\
			
			%---------------------
			
			{$x^7$}
			& & & & & & & 1.000
			& 0.434 
			& -0.083
			\\
			
			%---------------------
			{$x^8$}
			& & & & & & & & 1.000
			& 0.368
			\\
			
			%---------------------			
			{$x^9$}
			& & & & & & & & & 1.000
			\\
 
                %---------------------
                \midrule
			\rowcolor{white}
			{\footnotesize Variance}
			& {\multirow{2}{*}               {\tablenum{0.011}}} 
                & {\multirow{2}{*}{\tablenum{0.124}}} 
                & {\multirow{2}{*}{\tablenum{0.036}}} 
                & {\multirow{2}{*}{\tablenum{0.016}}} 
                & {\multirow{2}{*}{\tablenum{0.018}}}
                & {\multirow{2}{*}{\tablenum{0.044}}}
                & {\multirow{2}{*}{\tablenum{0.042}}}
                & {\multirow{2}{*}{\tablenum{0.047}}}
                & {\multirow{2}{*}{\tablenum{0.030}}}
			\\

                \rowcolor{white}
                {\footnotesize ($\times 10^{-2}$)}
                \\
   			
			%---------------------
			\bottomrule
		\end{tabular}
	}
	
\end{table}

\subsection{Incorporation of user's views}\label{sec:userviews}

In many applications of stochastic forward curve models, it is important to calibrate a model to market prices of derivatives or other information that may not be contained in the historical time series used in the estimation of the model parameters. This is particularly important in traditional risk-neutral models that are supposed to describe the future development of derivatives prices. Our aim is to develop a stochastic model (under the $P$-measure) that describes an agent's views concerning the future development of the involved risk factors. 

We apply the calibration procedure from Section~4.2 of \cite{aap21}, which allows us to incorporate both short-term forecasts as well as long-term views in both the macroeconomic model as well as the forward curve model described in Sections~\ref{sec:macro} and \ref{sec:termstruc} above. The calibration is achieved by choosing the sequence of vectors $a_t$ in \eqref{eq:var} so that the values of $x_t$ have desired median values at given points in time. Since both the macro model and the forward curve model are stationary, the model calibration according to Section~4.2 of \cite{aap21} is particularly easy. 

For this study, we choose the asymptotic median values for the macro model and the forward curve model according to the values given in Tables~\ref{table:ltmmacro} and~\ref{table:ltmxi}, respectively. In particular, the asymptotic median of inflation is set to $2\%$, which corresponds to the current central bank target. We set the asymptotic median of the FED target $L$ at $2.5\%$. This corresponds to the asymptotic median of $2.6\%$ for the transformed risk factor $\xi^1$.
Other values of the asymptotic medians are chosen simply as the historical median values of the corresponding time series. In addition to the asymptotic medians, we specify the short-term median values of the logarithmic target rate $\ln L$ so that the SOFR median will follow the last observed forward curve up to 5 years before converging towards the long-term median of $2.5\%$. It follows that the last observed forward curve sets the level of the simulated forward curve, whereas the historical median sets its shape. We emphasize that these choices are made only for purposes of illustration. It would be equally easy to use other specifications of the medians.

\begin{table}[ht]
	\centering
	\caption
	{
	\label{table:ltmmacro}
    Long-term median for the macroeconomic risk factors $y$. The $2.5\%$ for $L$ is user-specified. 
    The $2.00\%$ for $I$ corresponds to the inflation target set by the FOMC. The long-term median for $G$ is the median computed from historical data.}
	\resizebox{5cm}{!}
	{
            \tiny
            \begin{tabular}{lSSS}
			%---------------------
			\toprule
                & {$L$}
			& {$I$}
			& {$G$}  
			\\			
			%---------------------
                \midrule
			{$c$}
                & 2.50 \%
                & 2.00 \%
                & 2.85 \%
			\\
			%---------------------
			\bottomrule
		\end{tabular}
	}
\end{table}

\begin{table}[ht]
	\centering
	\caption
	{
		\label{table:ltmxi}
	Long-term medians for the term structure risk factors $\xi$. The median value of 0.026 for $\xi^1$ equals $\exp(m(x^1))+m(L)-c^1$, where $m(L)=0.025$ is the median value of $L$ from Table~\ref{table:ltmmacro} and $m(x^1)$ is the historical median of $x^1$. The remaining values were set equal to historical medians.
	}
	
	\resizebox{9cm}{!}
	{
            \begin{tabular}{lSSSSSSSSSS}
			%---------------------
			\toprule
			& {$\xi^1$}
			& {$\xi^2$}
			& {$\xi^3$}
			& {$\xi^4$}
			& {$\xi^5$} 
			& {$\xi^6$}
			& {$\xi^7$}
			& {$\xi^8$} 
			& {$\xi^9$}  
			\\
			
            %---------------------
                \midrule
               {$c$}
               & 0.026
               & 0.017
               & 0.016
               & 0.015
               & 0.014
               & 0.014
               & 0.015
               & 0.016
               & 0.012
			\\
			%---------------------
			\bottomrule
		\end{tabular}
	}
\end{table}

\section{A simulation study}\label{sec:simstudy}

This section illustrates the behaviour of the model built in Section~\ref{sec:tsm} by using numerical simulations. The linear structure of the underlying stochastic processes described in Sections~\ref{sec:macro} and~\ref{sec:termstruc} makes it easy to simulate the model over large numbers of scenarios on a daily frequency. We illustrate the simulations by numerically computing the development of the medians and confidence bands of the underlying risk factors. We will then use the model to study the payouts of interest rate derivatives in Section~\ref{sec:derpo}. Section~\ref{sec:ip} will compute indifference prices of the derivatives.

\subsection{Simulation of the macroeconomy and the forward curve}\label{sec:simexp}

We start by simulating 20K scenarios of the macroeconomic risk factors $y$ and the term structure risk factors $x$ over 10 years into the future. Along each scenario, we apply the inverse of the transformations in Section~\ref{sec:datatr} to obtain the values of the forward curve parameter~$\xi$. %The overall procedure is summarized in Algorithm~\ref{algo:sim} in Section~\ref{sec:concl}.

Figure~\ref{fig:simmacro} plots the medians, 95\% confidence bands, and an individual scenario of the macroeconomic risk factors. As expected, the median of $L$ converges to $2\%$ as specified by the constants in Table~\ref{table:ltmmacro}.
Figure~\ref{fig:simxi} presents analogous plots for the term structure risk factors $\xi$. Again, the medians converge to the values specified in Table~\ref{table:ltmxi}. Also, as specified in Section~\ref{sec:userviews}, the median of $\xi^1$ in Figure~\ref{fig:sim_sofr} tracks the forward curve observed at the start date of the simulation. One can also observe monthly jumps that correspond to changes in the FED policy rates. 
  
\begin{figure}[!ht]
    \centering
    { 
    \subfloat
    {
    \includegraphics[trim = 0mm 0mm 0mm 0mm, clip, width=0.45\textwidth] 
    {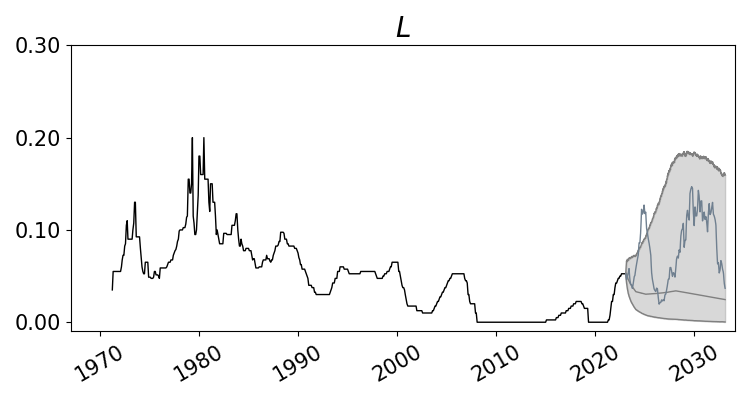}
    }
    \\[-1ex]
    \subfloat
    {
    \includegraphics[trim = 0mm 0mm 0mm 0mm, clip, width=0.45\textwidth] 
    {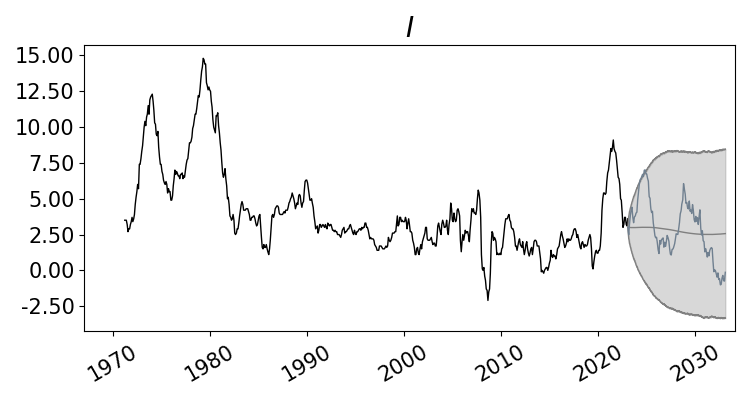}
    }
    \hfill
    \subfloat
    {
    \includegraphics[trim = 0mm 0mm 0mm 0mm, clip, width=0.45\textwidth] 
    {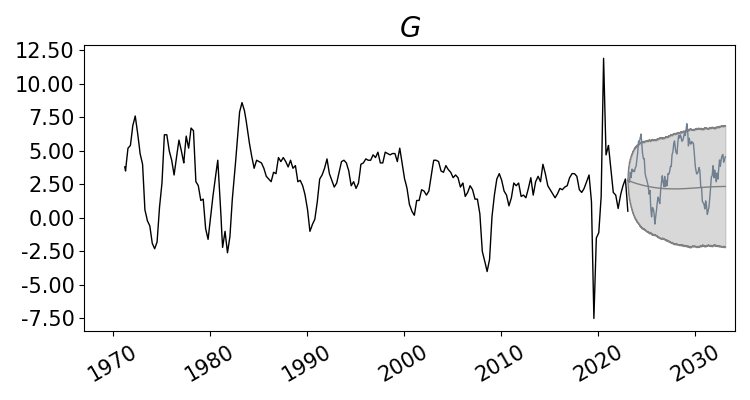}
    }
    }
    \centering
    \caption
    {
    Simulated scenarios for $L$ (in decimals), $I$ (in percentage), and $G$ (in percentage). The plots show the monthly historical data with a single simulated scenario, the median, and the 95\% confidence band.
    }
    \label{fig:simmacro}
\end{figure}

\begin{figure}[!ht]
    \centering
    \subfloat
    {
    \includegraphics[trim = 0mm 0mm 0mm 0mm, clip, width=0.45\textwidth] 
    {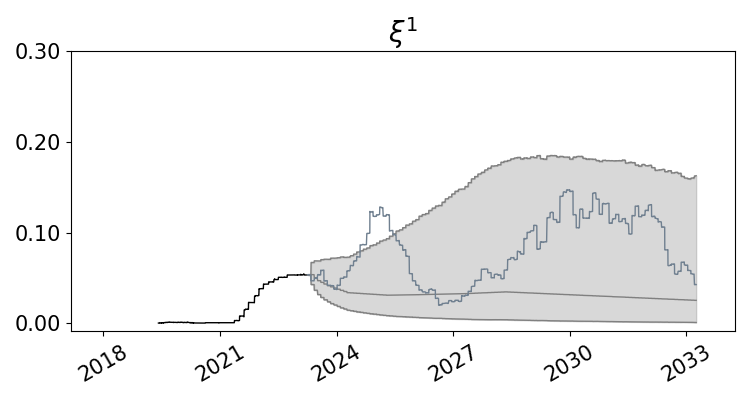}
    }
    \\[-1ex]
    \subfloat
    {
    \includegraphics[trim = 0mm 0mm 0mm 0mm, clip, width=0.45\textwidth] 
    {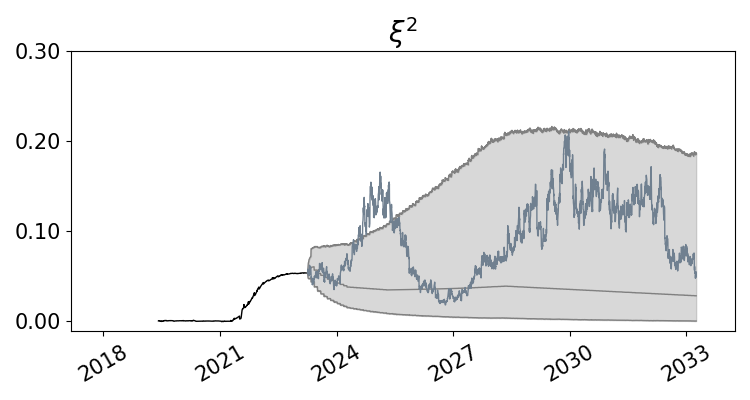}
    }
    \hfill
    \subfloat 
    {
    \includegraphics[trim = 0mm 0mm 0mm 0mm, clip, width=0.45\textwidth] 
    {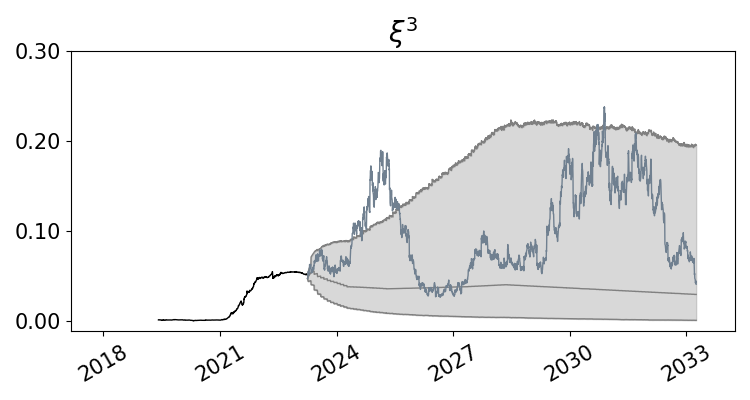}
    }
    \\[-1ex]
    \subfloat 
    {
    \includegraphics[trim = 0mm 0mm 0mm 0mm, clip, width=0.45\textwidth] 
    {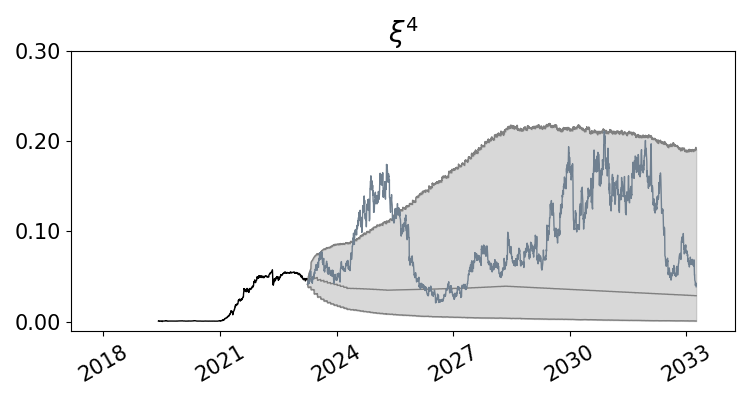}
    }
    \hfill
    \subfloat 
    {
    \includegraphics[trim = 0mm 0mm 0mm 0mm, clip, width=0.45\textwidth] 
    {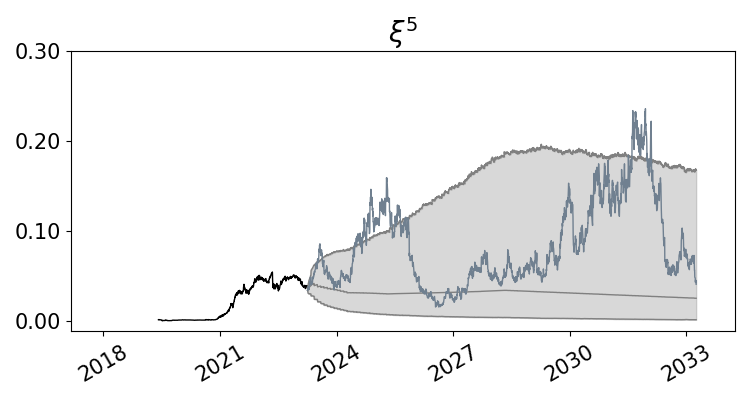}
    }
    \\[-1ex]
    \subfloat     {
    \includegraphics[trim = 0mm 0mm 0mm 0mm, clip, width=0.45\textwidth] 
    {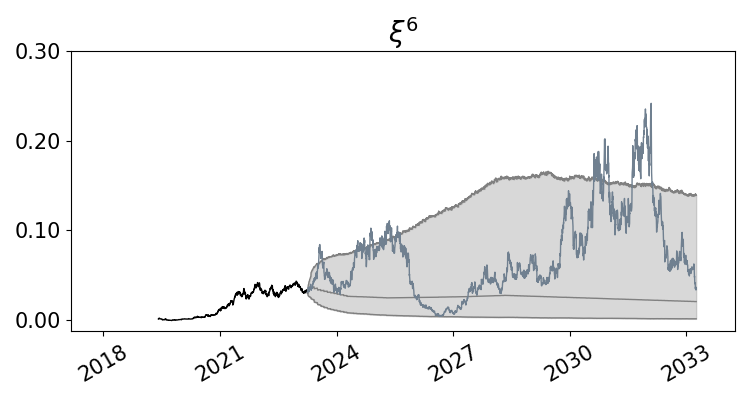}
    }
    \hfill
    \subfloat 
    {
    \includegraphics[trim = 0mm 0mm 0mm 0mm, clip, width=0.45\textwidth] 
    {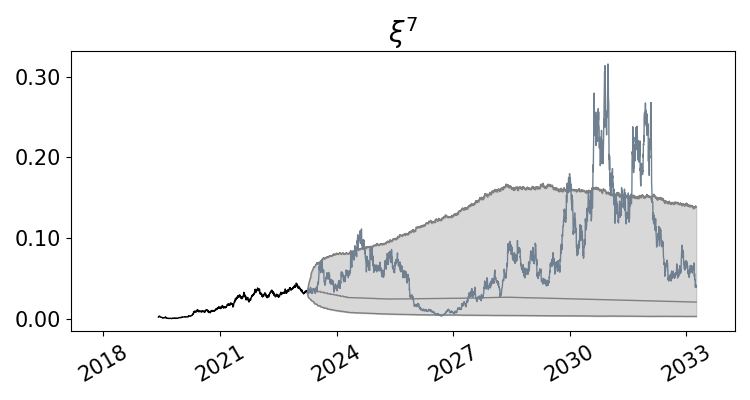}
    }
    \\[-1ex]
    \subfloat 
    {
    \includegraphics[trim = 0mm 0mm 0mm 0mm, clip, width=0.45\textwidth] 
    {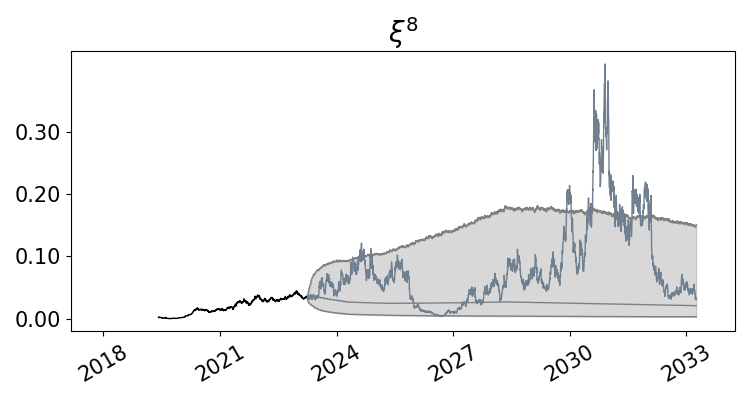}
    }
    \hfill
    \subfloat 
    {
    \includegraphics[trim = 0mm 0mm 0mm 0mm, clip, width=0.45\textwidth] 
    {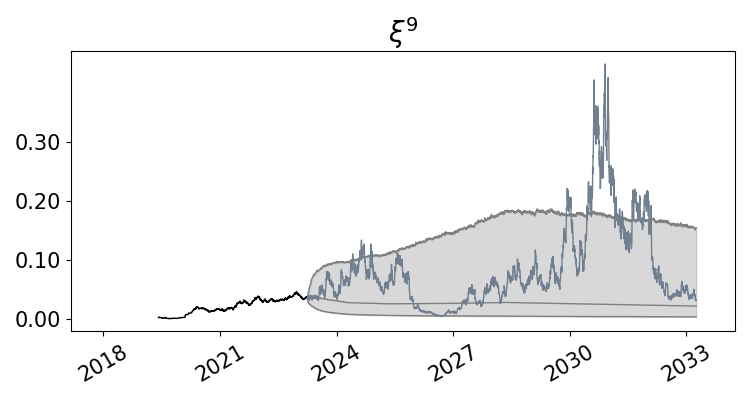}
    }
    \centering
    \caption
    {
    Simulated scenarios for $\xi^k$, $k=2,\ldots,9$. The plots show the daily historical data with a single simulated scenario, the median, and the 95\% confidence band.
    }
    \label{fig:simxi}
\end{figure}

\begin{figure}[!ht]
    \centering
    \subfloat
    {
    \includegraphics[trim = 0mm 0mm 0mm 0mm, clip, width=0.6\textwidth] 
    {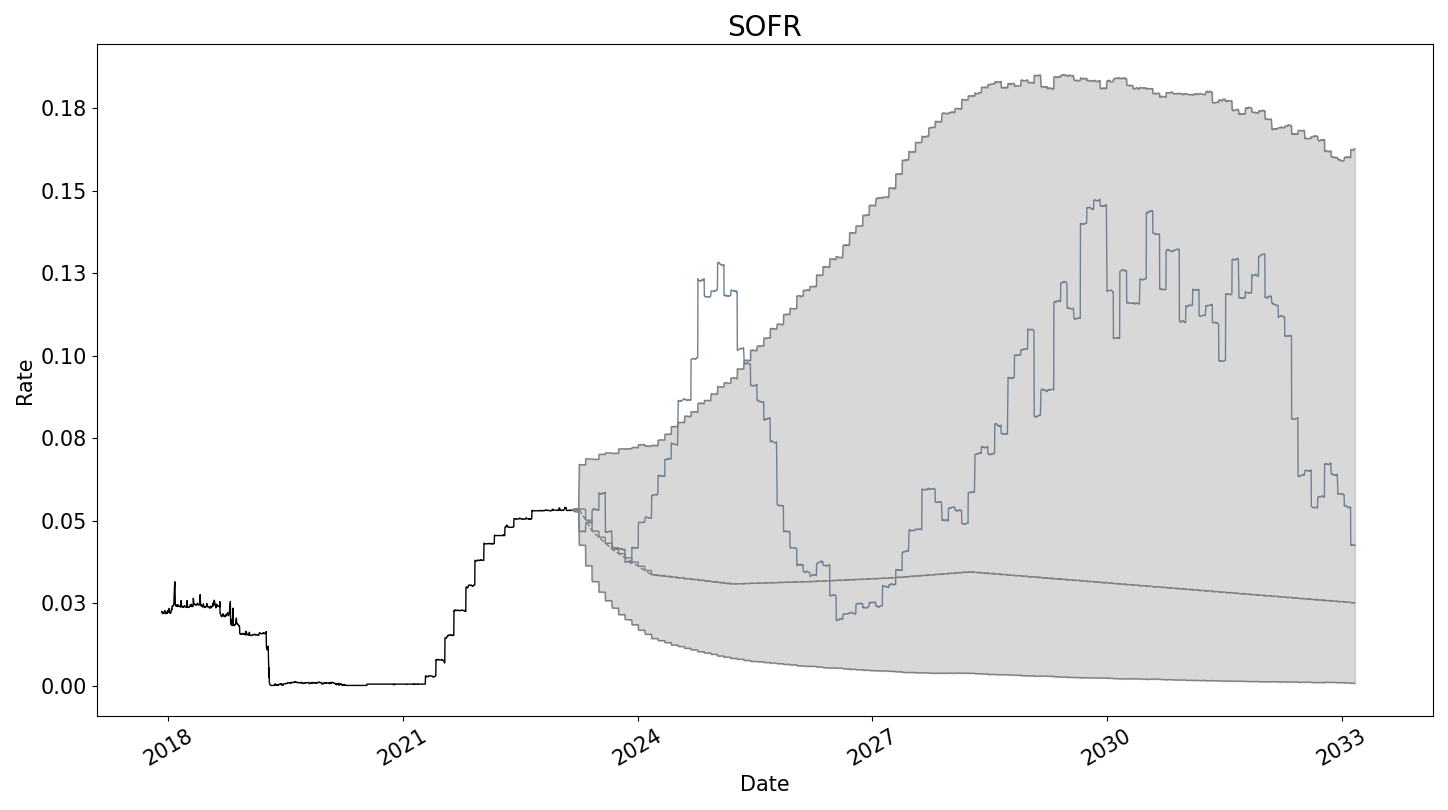}
    }
    \\[-1ex]
    \subfloat
    {
    \includegraphics[trim = 0mm 0mm 0mm 0mm, clip, width=0.6\textwidth] 
    {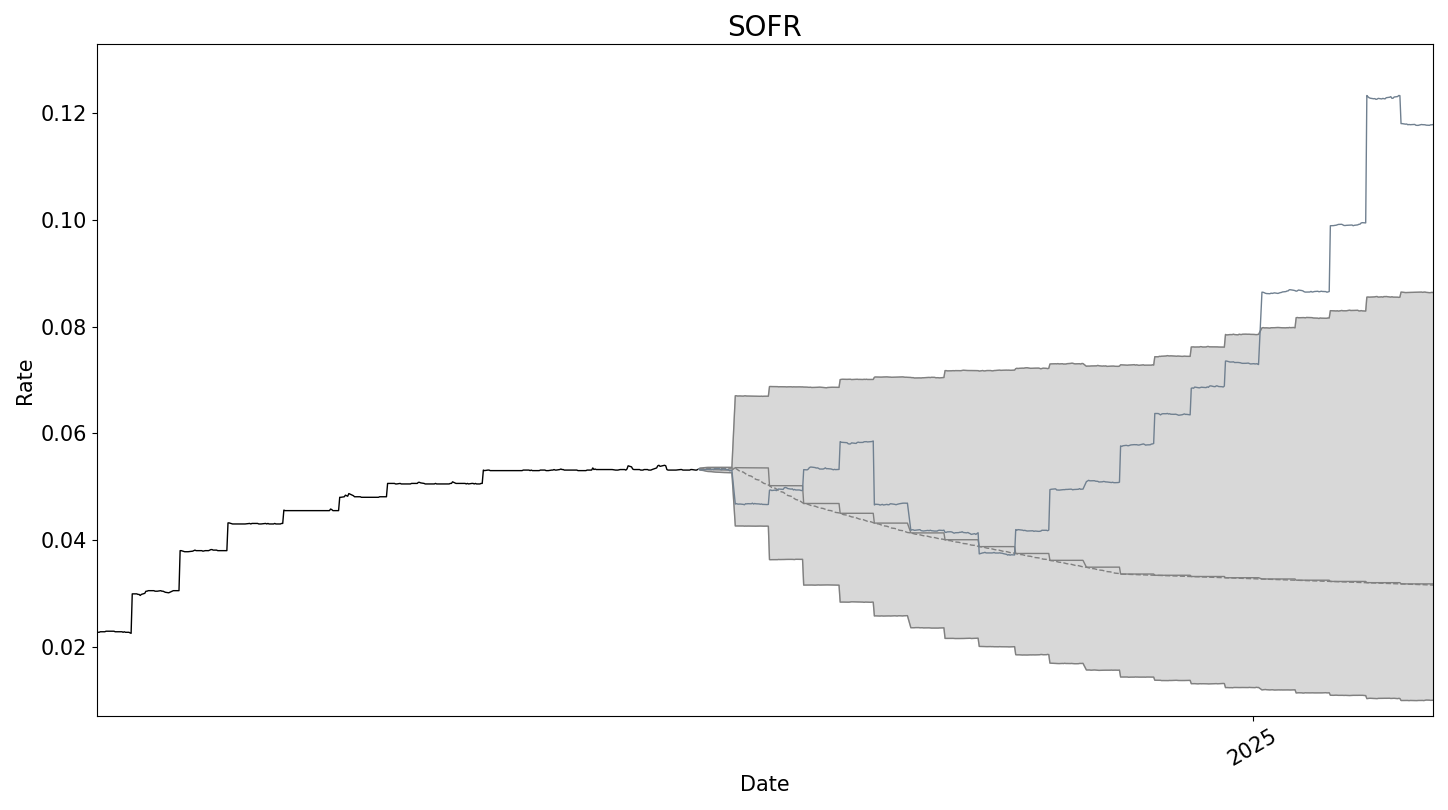}
    }
    \centering
    \caption
    {
    Simulated scenarios for SOFR are shown in the plots. The upper plot depicts 10 years of daily historical data alongside a single simulated scenario, its median, and a 95\% confidence band. The lower plot zooms in on the period from 2022 to 2026, highlighting the last observed forward curve (dashed grey line) and the jumps in the simulated scenario (solid blue line).
    }
    \label{fig:sim_sofr}
\end{figure}

\subsection{Simulating SOFR derivatives payouts}\label{sec:derpo}

This section uses the forward curve simulations from Section~\ref{sec:simexp} to simulate the payouts of selected SOFR derivatives as of August 28, 2024, with maturities extending beyond this date. The derivatives and their payout formulas are given in Table~\ref{table:po}. All the quantities appearing in the formulas can be calculated in each scenario from the simulated overnight forward curves using the formulas given in Section~\ref{sec:forward}. In particular, the SOFR average $R(t_0,t_1)$ in Table~\ref{table:po} is given by formula \eqref{eq:SOFR_Average}, where the SOFR rate $r_t$ is given by 
\[
r_t=F_t(0,1)=\frac{e^{F_t(0)\delta}-1}{\delta};
\]
see \eqref{sofrf}. The zero-coupon bond prices $P_T(T+t_k)$ are given by \eqref{zcb} while the futures rates $F_{t_1-1}(t_1,t_2)$ are given by \eqref{eq:FF}.

\begin{table}[ht]
	\centering
	\caption
		{\label{table:po}
		Payouts in USD for long positions in SOFR derivatives. The payouts follow the CME convention where the quantities in the formulas appear to have flipped signs; see Remark~\ref{rem:imm}. The maturities are $t_1$ for the futures, $t < t_1$ for the call and put, while the maturity of the swaption is denoted by $T$. The strikes are denoted by $X$.
		}
		\resizebox{7cm}{!}
		{
                \renewcommand{\arraystretch}{2}
                \begin{tabular}{lll}
				\toprule
                    {Derivative}
				& {Payout in USD for 1 contract}
				\\
				\midrule
				{Three-month SOFR futures}
				& {$[F_t^b(t_0,t_1)-R(t_0,t_1)]\times 250000$}
		          \\
				\rowcolor{mygray}
				{Call on three-month SOFR futures}
				& {$ P_{t}(t_2)(X - F_{t}(t_1, t_2))^{+}  \times  250000$}
				\\
				{Put on three-month SOFR futures}
				& {$ P_{t}(t_2)(F_{t}(t_1, t_2) - X)^{+}  \times  250000$}
                    \\
                    \rowcolor{mygray}
				{SOFR call swaption}
				& {$\left(X\sum_{k=1}^K P_T(T+t_k) -1 + P_T(T+t_K)\right)^+$}
                    \\
				\bottomrule
			\end{tabular}
		}

\end{table}

If exercised, a swaption, with strike $X$ and maturity $T$, converts to a short position in an OIS with swap rate $X$ and $K$ yearly payments; see~\cite{CMEIRprod, CMEswaption}. Assuming perfectly liquid zero-coupon bonds, the market price of such an OIS would be
\[
X \sum_{k=1}^{K} P_T(T+t_k) \delta (t_{k} - t_{k-1}) - [P_T(T) -P_T(T+t_K) ].
\]
Indeed, the first term is the market price of the outstanding fixed-leg payments, while the second term is the replication cost of the floating leg calculated like that of the futures contract in Section~\ref{sec:forward}. Clearly, $P_T(T)=1$. The swaption payout in Table~\ref{table:po} is thus the positive part of the liquidation value of the swaption at maturity. If the liquidation value is negative, a rational agent would not exercise the option.

SOFR futures options deliver the underlying futures contract when exercised. If the underlying futures market is perfectly liquid, we can assume that the options are cash-settled. Indeed, taking the opposite futures position at the time of exercise would yield a cash payout of 
\begin{equation*}
    P_t(t_2)(X - F_t(t_1, t_2))^+ \times 250000,
\end{equation*}
where $X$ is the option strike and $F_t(t_1, t_2)$ is the futures price at the time $t$ of exercise.

Figure~~\ref{fig:deriv_po_dens} displays the kernel density plots of the payouts of the four derivatives.
The swaption has strike $X=0.033769$, maturity $T=5$ years, and underlying OIS with $K=5$ yearly payments.
The futures contract has maturity $t_1$, the third Wednesday of December 2024, and the bid-rate $F^b_t(t_0, t_1) = 0.0528$, with $t_0$ the third Wednesday of September 2024. The call and the put both have strike $X = 0.045$ and maturity $t_1 - 1$. 
The end of the reference period $t_2$ for the underlying $F_{t_1 - 1}(t_1, t_2)$ is on the third Wednesday of March 2025.
Note that these data are based on observed market data as of August 28, 2024.

\begin{figure}[!ht]
    \centering
    \subfloat
    {
    \includegraphics[trim = 0mm 0mm 0mm 0mm, clip, width=0.45\textwidth] 
    {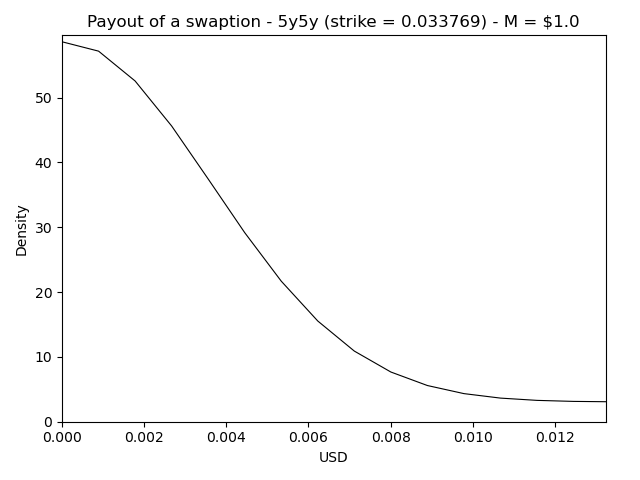}
    }
    \hfill
    \subfloat
    {
    \includegraphics[trim = 0mm 0mm 0mm 0mm, clip, width=0.45\textwidth] 
    {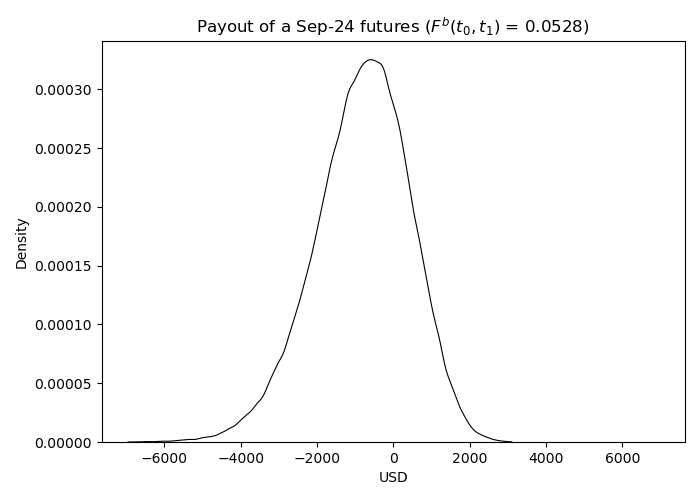}
    }
    \\[-1ex]
    \subfloat
    {
    \includegraphics[trim = 0mm 0mm 0mm 0mm, clip, width=0.45\textwidth] 
    {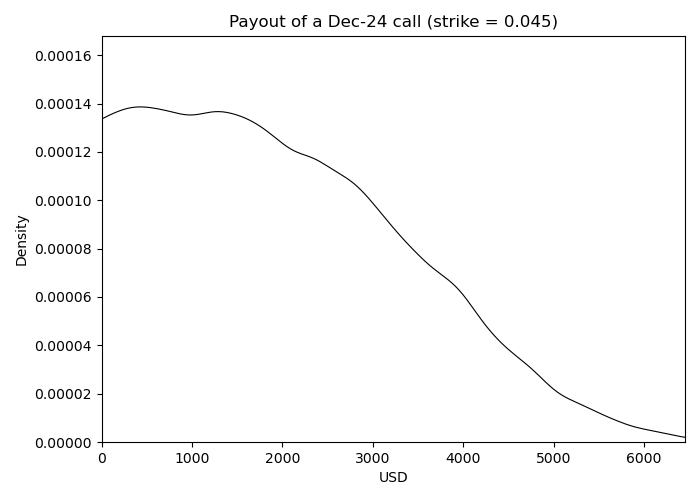}
    }
    \hfill
   \subfloat
    {
    \includegraphics[trim = 0mm 0mm 0mm 0mm, clip, width=0.45\textwidth] 
    {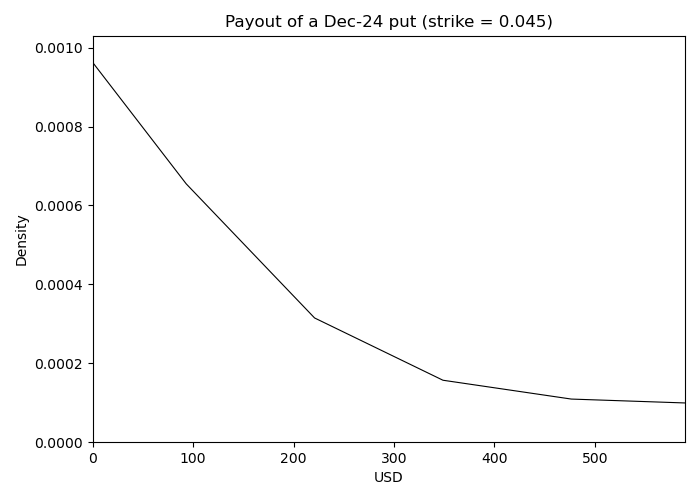}
    }
    \centering
    \caption
    {
    Payout density of SOFR derivatives as of 28 August 2024.
    }
    \label{fig:deriv_po_dens}
\end{figure}

\subsection{Indifference pricing of SOFR derivatives}\label{sec:ip}

A common approach to derivative pricing is to construct a market model together with a martingale measure which is then used for pricing derivatives by taking expectations of their cashflows. Fixing a pricing measure, however, leads to a linear pricing rule, which contradicts the nonlinearities observed in practice. Such problems are addressed explicitly by the techniques of incomplete markets, where one first develops a term structure model under the \enquote{statistical} measure and then defines prices according to the indifference principle; see e.g.~\cite{car9}, \cite{pen14} and their references. This section uses the model developed in Sections~\ref{sec:str}, ~\ref{sec:par}, and~\ref{sec:tsm} in indifference pricing of the derivatives described in the previous section. The exercise below is quite simple in that the only \enquote{hedging instrument} is the money market account represented by the overnight Repo market. More sophisticated hedging strategies will be built in a follow-up paper that uses exchange-traded derivatives to hedge OTC options. %takes the stochastic model developed in this paper as the description of the underlying risk factors.
 
Let $c$ be the payout at maturity $t_1$ of a SOFR derivative and $U:\reals\to\reals$ a concave strictly increasing utility function. The utility indifference selling price at $t_0$ is the least amount of cash $\alpha$ for which an agent, with initial wealth $\bar{w}$, is indifferent between selling the claim $c$ for $\alpha$, and not doing the trade. In the case of the futures, where the price is paid at maturity, the indifference price is the solution $\alpha$ to the equation 
\begin{equation}\label{eq:exputil}
    E U\left(\bar{w} \prod_{t=t_0}^{t_1-1} (1+r_t \delta) + \alpha - c \right) = E U \left(\bar{w}  \prod_{t=t_0}^{t_1-1} (1+r_t \delta)\right).
\end{equation}
In the numerical examples below, we use the exponential utility function $U(c):=\frac{1}{\rho} \left( 1 - e^{-\rho c} \right)$ where $\rho >0$ is a risk aversion parameter. Taking logarithms on both sides of \eqref{eq:exputil} gives 
% the entropic risk measure
\begin{equation}\label{eq:sellprfutswapt}
    \alpha = \frac{1}{\rho} \ln\frac{E \exp \left\{ -\rho \left[ \bar{w}  \prod_{t=t_0}^{t_1-1} (1+r_t \delta) - c \right] \right\}}{E \exp \left\{ -\rho \bar{w} \prod_{t=t_0}^{t_1-1} (1+r_t \delta)  \right\}}.
\end{equation}

For the call, put, and swaption, where the price is received at the time of purchase, we use the SOFR overnight roll-over account as a numeraire. This allows for the computation of the selling price again analytically without numerical line search. Indeed, the selling price is the solution $\alpha$ to the equation
\begin{equation}\label{eq:exputilopt}
    EU\left(\bar{w} + \alpha - \frac{c}{\prod_{t=t_0}^{t_1-1} (1+r_t \delta)} \right) = EU\left(\bar{w} \right).
\end{equation}
With the exponential utility, we obtain
\begin{equation}\label{eq:optprice}
    \alpha = \frac{1}{\rho} \ln E \exp \left \{ \frac{\rho c}{\prod_{t=t_0}^{t_1-1} (1+r_t \delta)}\right \}.
\end{equation}

The indifference price $\alpha$ in \eqref{eq:sellprfutswapt} 
and \eqref{eq:optprice} 
can be computed by numerically evaluating the expectations. In the numerical computations below, we will use Monte Carlo with antithetic sampling. That is, we simulate a random sample of $N$ scenarios and obtain another $N$ scenarios by reflecting the underlying Gaussian process.

Figures~\ref{fig:deriv_price_surf} and~\ref{fig:sofr_options_rho}  plot the indifference selling prices for the four SOFR derivatives studied in Section~\ref{sec:derpo}.
Figure~\ref{fig:deriv_price_surf} displays the indifference price of the futures as a function of the initial wealth $\bar w$ and the risk aversion parameter $\rho$.
As $\bar{w}$ increases, the price increases. With an initial financial position of cash only, the agent is less desperate for extra cash and asks for a higher price to compensate for taking on the additional risk of selling the derivative.
As $\rho$ increases, the price increases. Because their initial financial position consists only of cash, the higher the value of $\rho$ the more risk-averse the agent is. They require more money to sell the derivative. This is also observed in the plots of Figure~\ref{fig:sofr_options_rho} which display the swaption and option price as functions of the risk aversion parameter $\rho$.

\begin{figure}[!ht]
    \centering
    {
    \includegraphics[trim = 0mm 0mm 0mm 0mm, clip, width=0.55\textwidth] 
    {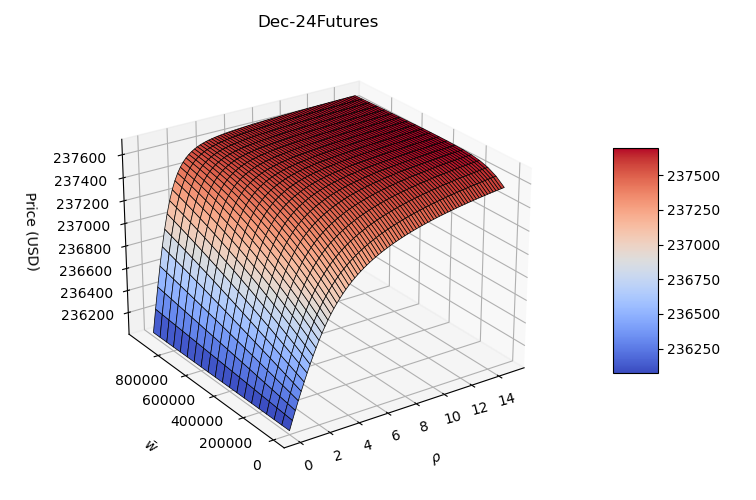}
    }
    \centering
    \caption
    {
    Selling price of SOFR three-month futures as a function of the initial wealth $\bar{w}$ and the risk aversion parameter $\rho$ (scaled by a factor of $10^{-3}$) as of 28 August 2024.
    }
    \label{fig:deriv_price_surf}
\end{figure}

\begin{figure}[!ht]
    \centering
    \subfloat
    {
    \includegraphics[trim = 0mm 0mm 0mm 0mm, clip, width=0.45\textwidth] 
    {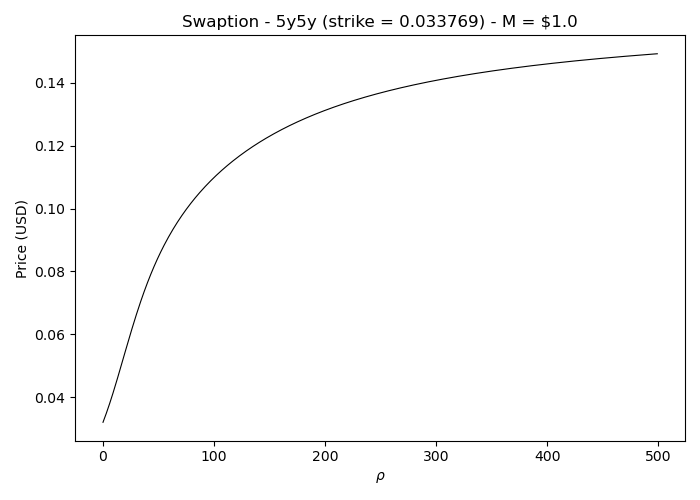}
    }
    \\
    \subfloat
    {
    \includegraphics[trim = 0mm 0mm 0mm 0mm, clip, width=0.45\textwidth] 
    {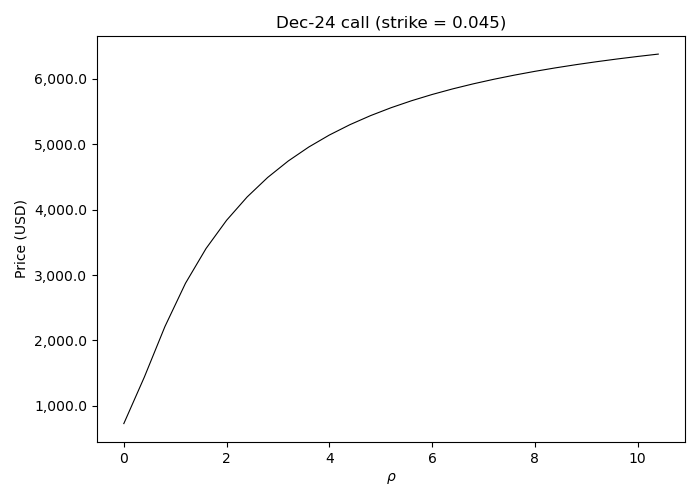}
    }
    \hfill
    \subfloat
    {
    \includegraphics[trim = 0mm 0mm 0mm 0mm, clip, width=0.45\textwidth] 
    {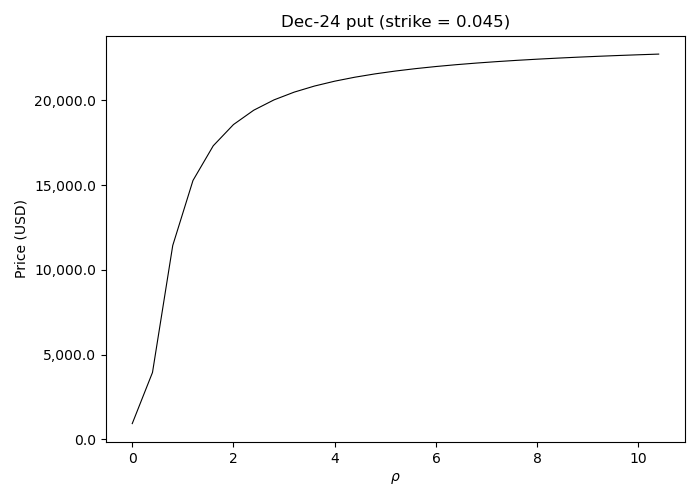}
    }
    \centering
    \caption
    {
    Selling price of SOFR swaption and options on three-month futures as a function of the risk aversion parameter $\rho$ (scaled by a factor of $10^{-3}$) as of 28 August 2024.
    }
    \label{fig:sofr_options_rho}
\end{figure}

\section{Arbitrage}\label{sec:arb}

This section takes a brief look at no-arbitrage conditions on discrete-time term structure models.
The analysis is based on Theorem~\ref{thm:js} below which provides a convenient no-arbitrage condition for perfectly liquid market models in discrete time. It is concerned with an abstract model of a financial market where $d+1$ perfectly liquid securities can be traded over discrete finite time. 

We will denote the unit price of asset $j=0,\ldots,d$ at time $t=0,\ldots,T$ by $s^j_t$. We will assume that the price of asset $j=0$ is always strictly positive and denote the \enquote{discounted} prices by
\[
S_t^j:=s^j_t/s^0_t.
\]
In particular, $S_t^0\equiv 1$. It is then easily shown that the discounted terminal wealth of a self-financing trading strategy $x=(x_t)_{t=0}^{T-1}$ can be expressed as the discrete-time stochastic integral
\[
V_T(x):=\sum_{t=0}^{T-1}x_t\cdot\Delta S_{t+1},
\]
where $S_t=(S_t^j)_{j=1}^d$, $\Delta S_t:=S_t-S_{t-1}$ and $x_t=(x_t^j)_{j=1}^J$ denotes the portfolio of \enquote{risky assets} $j=1,\ldots,d$ held over period $(t,t+1]$. As usual, we model the trader's uncertainties and information by a probability space $(\Omega,\F,P)$ and a filtration $\FF$, respectively. The linear space of adapted trading strategies $x=(x_t)_{t=0}^{T-1}$ will be denoted by $\N$. The market model is {\em arbitrage-free} if every trading strategy $x\in\N$ with $V_T(x)\ge 0$ almost surely satisfies $V_T(x)=0$ almost surely.

The following theorem is from \cite[Theorem~3]{js98}, but one should note that part g of the statement there has a typo: the symbol \enquote{$Q_n$} should be replaced by \enquote{$\bar Q_n$}. An alternative proof and an extended statement can be found in~\cite[Remark~2.120]{pp24}. The {\em ${\F_t}$-conditional support} of a random variable $X\in L^0(\Omega,\F,P;\reals^d)$ is the smallest closed-valued $\F_t$-measurable mapping $\supp_t X:\Omega\tos\reals^d$ such that $X\in\supp_t X$ almost surely. The {\em convex hull} and the {\em relative interior} of a set $C\subseteq\reals^d$ will be denoted by $\co C$ and $\ri C$, respectively.

\begin{theorem}[Jacod and Shiryaev]\label{thm:js}
A price process $s$ is arbitrage-free if and only if the discounted price process $S$ satisfies
\[
S_t\in\ri\co\supp_t S_{t+1}
\]
almost surely for every $t=0,\ldots,T-1$.
\end{theorem}

\subsection{Arbitrage in zero-coupon bonds}\label{sec:arbzcb}

In the literature of mathematical term structure models, it is common to assume that zero-coupon bonds of all maturities can be traded dynamically over time without frictions. No-arbitrage conditions then impose strong consistency conditions on the development of the zero-coupon bond prices and the forward curve. This section takes a brief look at what kind of conditions this would impose on the forward curve dynamics in the discrete time setting.

Consider a term structure model where the future forward curves
\[
F_t := (F_t(s))_{s=0}^{\bar T}
\]
are modeled as $\F_t$-measurable random vectors. Recall from Section~\ref{sec:forward} that, in perfectly liquid markets, the forward curve is related to prices of zero-coupon bonds by
\begin{equation}\label{eq:zcbd}
P_t(T) = e^{-\sum_{s=0}^{T-1}F_t(s)\delta}.
\end{equation}

\begin{theorem}
Assume that, at each time $t$, zero-coupon bonds of all maturities up to $\bar T$ days ahead can be traded without frictions at prices given by \eqref{eq:zcbd}. The market is arbitrage-free if and only if 
\[
H((F_t(s))_{s=1}^{\bar T})\in\ri\co H\left(\supp_t(F_{t+1}(s))_{s=0}^{\bar T-1}\right),
\]
where the function $H:\reals^{\bar T}\to\reals^{\bar T}$ is given by $H=\exp\circ L$, where the exponential is applied componentwise and 
\[
L=-\delta\begin{bmatrix}
1 & 0 & & 0\\
1 & 1 & & 0\\
\vdots & & & \\
1 & 1&\cdots & 1
\end{bmatrix}.
\]
\end{theorem}

\begin{proof}
Assuming perfect liquidity, one unit of cash invested at time $t$ in the zero-coupon bond with maturity $t+T+1$ yields 
\[
\frac{P_{t+1}(T)}{P_t(T+1)}
\]
units of cash at time $t+1$. The stochastic process $s^T$ defined recursively by $s_0^T=1$ and 
\[
s^T_{t+1} = s^T_t\frac{P_{t+1}(T)}{P_t(T+1)}\quad t=1,2, \ldots
\]
describes the development of the value of a self-financing investment strategy that, at each time $t$, reinvests everything in the zero coupon bond with maturity $t+T+1$. In particular, the process $s^0$ describes investments in the overnight market:
\[
s^0_{t+1} = s^0_t\frac{1}{P_t(1)}\quad t=1,2, \ldots
\]
Taking $s^0$ as the numeraire, the discounted prices become
\begin{equation}\label{eq:S}
S^T_{t+1}:=\frac{s^T_{t+1}}{s^0_{t+1}} = S^T_t\frac{P_{t+1}(T)}{P_t(T+1)/P_t(1)}.
\end{equation}

Denote $S_t:=(S_t^T)_{T=1}^{\bar T}$. By Theorem~\ref{thm:js}, the market with perfectly liquid zero-coupon bonds is arbitrage-free if and only if $S_t\in\ri\co\supp_t S_{t+1}$ for all $t$ almost surely. Using \eqref{eq:S}, this can be written as
\[
(S_t^T)_{T=1}^{\bar T}
\in\ri\co\supp_t \left(S^T_t\frac{P_{t+1}(T)}{P_t(T+1)/P_t(1)}\right)_{T=1}^{\bar T}
\]
or, equivalently,
\[
\left(\frac{P_t(T+1)}{P_t(1)}\right)_{T=1}^{\bar T}\in\ri\co\supp_t (P_{t+1}(T))_{T=1}^{\bar T}.
\]
By \eqref{eq:zcbd},
\[
\left(\frac{P_t(T+1)}{P_t(1)}\right)_{T=1}^{\bar T} = \left(e^{-F_t(1)\delta},\ldots,e^{-\sum_{s=1}^{\bar T}F_t(s)\delta}\right) = H((F_t(s))_{s=1}^{\bar T})
\]
and
\[
(P_{t+1}(T))_{T=1}^{\bar T} = \left(e^{-F_{t+1}(0)\delta},\ldots,e^{-\sum_{s=0}^{\bar T}F_{t+1}(s)\delta}\right) = H\left((F_{t+1}(s))_{s=0}^{\bar T-1}\right).
\]
The no-arbitrage condition can thus be written as
\[
H((F_t(s))_{s=1}^{\bar T})\in\ri\co\supp_t H\left((F_{t+1}(s))_{s=0}^{\bar T-1}\right).
\]
By \cite[Theorem~2.43.4]{pp24}, the continuity of $H$ implies
\[
\supp_t H\left((F_{t+1}(s))_{s=0}^{\bar T-1}\right) = \cl H\left(\supp_t(F_{t+1}(s))_{s=0}^{\bar T-1}\right).
\]
The claim now follows from the fact that $\ri\co\cl C=\ri\co C$ for any set $C$ in a Euclidean space; see \cite[Theorem~6.3]{roc70a}.
\end{proof}

Note that $H$ maps $\reals^{\bar T}_{++}$ to
\[
D:=\{p\in\reals^{\bar T}\mid 1>p_1>p_2>\cdots>p_{\bar T}>0\}.
\]
The no-arbitrage condition thus holds, in particular, if 
\[
\supp_t(F_{t+1}(s))_{s=0}^{\bar T-1}=\reals^{\bar T}_+
\]
i.e.\ if the conditional support of the day-ahead forward curve contains every positive forward curve. This condition is far from necessary but it could be enforced e.g.\ by multiplying each $F_t(s)$ by $e^{\epsilon_t}$, where $\epsilon_t$ are independent Gaussian variables with nonzero variance. Such additional randomness could be interpreted as observation errors, as is common in the theory of affine term structure models; see e.g.~\cite[Section~4]{cheridito2007market}. Taking the variance of $\epsilon_t$ sufficiently close to zero, the additional randomness would not have a material effect on portfolio optimization or indifference pricing in the face of bid-ask spreads. More importantly, one cannot trade perfectly liquid zero-coupon bonds with arbitrary maturities in practice, so the no-arbitrage conditions derived above aren't all that relevant in real-life applications.

\subsection{Arbitrage in futures market}

Unlike zero-coupon bonds, futures contracts are exchange-traded and fairly liquid for many reference periods. It is thus natural to ask what a forward curve model should satisfy to preclude arbitrage opportunities when trading futures. 

Consider a market with a finite number $J$ of futures contracts each one $j=1,\ldots, J$ for a given reference period $(t_0^j,t_1^j)$. For simplicity, we assume that the futures are perfectly liquid so that, both the bid and ask quotes for a futures contract for reference period $(t^j_0,t^j_1)$ are given for all times $t<t^j_0$ by equation \eqref{eq:FF}, i.e.
\begin{equation}\label{eq:frates}
F_t(t^j_0,t^j_1) = \frac{1}{(t^j_1-t^j_0)\delta}\Bigg[\exp\Bigg(\sum_{s=t^j_0-t}^{t^j_1-t-1}F_t(s)\delta\Bigg)-1\Bigg].
\end{equation}

Under the assumption of perfect liquidity, one can enter an arbitrary futures position $z^j_t\in\reals$ at time $t$ without a cost and the position would yield 
\[
z^j_t[R(t^j_0,t^j_1)-F_t(t^j_0,t^j_1)]
\]
units of cash at maturity $t^j_1$; see Section~\ref{sec:forward}. A dynamic trading strategy $z^j=(z^j_t)_{t=0}^{T}$ would yield
\begin{align}
C^j(z^j):=&\sum_{t=0}^{T} z^j_t[R(t^j_0,t^j_1)-F_t(t^j_0,t^j_1)]\nonumber\\
=&R(t^j_0,t^j_1)\sum_{t=0}^{T} z^j_t-\sum_{t=0}^{T} z^j_tF_t(t^j_0,t^j_1)\label{eq:C}
\end{align}
units of cash at maturity. We assume that $T<t^j_1$ and that the futures position is closed at $T$ so that
\[
\sum_{t=0}^{T} z^j_t = 0
\]
and the first term in \eqref{eq:frates} drops out. The payout of a trading strategy $z\in\N$ in the $J$ futures contracts is given by
\[
C(z):=\sum_{j=1}^JC_j(z^j).
\]
The no-arbitrage condition means that, if $z\in\N$ is such that $C(z)\ge 0$ almost surely, then $C(z)=0$ almost surely.

\begin{theorem}\label{thm:arbfut}
Assume that, at each time $t$,  futures contracts for reference periods $(t^j_0,t^j_1)$, $j=1,\ldots,J$ can be traded without bid-ask spreads at futures rates given by \eqref{eq:frates}. The market is arbitrage-free if and only if
\[
\exp((G_t^j)_{j=1}^J)\in\ri\co\exp\left(\supp_t(G_{t+1}^j)_{j=1}^J\right),
\]
where
\[
G_t^j := \sum_{s=t^j_0-t}^{t^j_1-t-1}F_t(s)\delta.
\]
\end{theorem}

\begin{proof}
Defining
\[
x_t:=\sum_{t=0}^tz_t,
\]
we have
\[
C(x)=-\sum_{j=1}^J\sum_{t=0}^{T} F_t(t^j_0,t^j_1)\Delta x^j_t.
\]
Rearranging terms (integrating by parts) gives
\[
C(x)=\sum_{j=1}^J\sum_{t=0}^{T-1}x^j_t\Delta F_{t+1}(t^j_0,t^j_1).
\]
Expression \eqref{eq:frates} gives
\[
\Delta F_{t+1}(t^j_0,t^j_1) = \frac{1}{(t^j_1-t^j_0)\delta}\Delta S^j_{t+1},
\]
where
\[
S^j_t:=\exp\Bigg(\sum_{s=t^j_0-t}^{t^j_1-t-1}F_t(s)\delta\Bigg).
\]
The no-arbitrage condition can thus be written in the standard form
\[
x\in\N,\ \sum_{t=0}^{T-1}x_t\cdot\Delta S_{t+1}\ge 0\ a.s.\quad\Longrightarrow\quad \sum_{t=0}^{T-1}x_t\cdot\Delta S_{t+1}=0.
\]
By Theorem~\ref{thm:js}, this holds if and only if
\begin{equation}\label{eq:rico}
S_t\in\ri\co\supp_tS_{t+1}
\end{equation}
almost surely. This can be written as
\[
(\exp(G_t^j))_{j=1}^J\in\ri\co\supp_t(\exp(G_{t+1}^j))_{j=1}^J,
\]
where
\[
G_t^j := \sum_{s=t^j_0-t}^{t^j_1-t-1}F_t(s)\delta.
\]
By \cite[Theorem~2.43.4]{pp24}, continuity of the exponential function implies 
\[
\supp_t(\exp(G_{t+1}^j))_{j=1}^J = 
\cl\exp\left(\supp_t(G_{t+1}^j)_{j=1}^J\right).
\]
The claim now follows from the fact that $\ri\co\cl C=\ri\co C$ for any set $C$ in a Euclidean space; see \cite[Theorem~6.3]{roc70a}.
\end{proof}

A simple sufficient condition for the no-arbitrage criterion in Theorem~\ref{thm:arbfut} is that
\[
\exp((G_t^j)_{j=1}^J)\in\inte\exp\left(\supp_t(G_{t+1}^j)_{j=1}^J\right),
\]
where \enquote{$\inte$} stands for the interior in the Euclidean topology. By continuity of the exponential, this condition can be written equivalently as
\begin{equation}\label{eq:Gints}
(G_t^j)_{j=1}^J\in\inte\supp_t(G_{t+1}^j)_{j=1}^J.
\end{equation}

Assume now that 
\[
F_t(s) = \sum_{k=1}^K\xi_t^k\phi^k(s-t)\quad s\ge t
\]
as in \eqref{parf}. We get
\begin{align*}
G_t^j &= \sum_{s=t^j_0-t}^{t^j_1-t-1}F_t(s)\delta = \sum_{s=t^j_0-t}^{t^j_1-t-1}\sum_{k=1}^K\xi_t^k\phi^k(s-t)\delta =  \sum_{k=1}^K\xi_t^k\Phi^{j,k}_t
\end{align*}
where
\[
\Phi^{j,k}_t := \sum_{s=t^j_0-t}^{t^j_1-t-1}\phi^k(s-t)\delta.
\]
We can express $G_t=(G^j_t)_{j=1}^J$ in the matrix notation as $G_t=\Phi_t\xi$, where $\Phi_t$ is the $J\times K$-matrix with entries $\Phi^{j,k}_t$. Thus, condition \eqref{eq:Gints} can be written as
\[
\Phi_t\xi_t\in\inte\supp_t \Phi_{t+1}\xi_{t+1}
\]
or, equivalently,
\[
\Phi_t\xi_t\in\inte\Phi_{t+1}\supp_t\xi_{t+1}.
\]

The above conditions should not, however, be taken too seriously since they are based on the assumption of perfect liquidity. The conditions may amount to excessive restrictions that do not allow for phenomena observed in real markets where futures contracts are subject to bid-ask spreads and margin requirements. One could also add small \enquote{measurement errors}, as discussed at the end of Section~\ref{sec:arbzcb}, to remove possible arbitrage opportunities.

\section{Conclusion}

This paper presents a statistical SOFR term structure model, which is easily calibrated to historical data, current market conditions as well as users' views. The model operates on a daily frequency, and it is consistent with the discrete jumps often observed on FMOC meeting dates. The model allows for fast simulations, which are needed in e.g.\ in risk management, portfolio optimization and indifference pricing of interest rate derivatives.

The calibration of the model breaks down into three separate steps summarized in Algorithms~\ref{algo:estimatexi} to~\ref{algo:spreadcalib} below. In the first step, Algorithm~\ref{algo:estimatexi}
fits a parameterized forward curve to observed daily snapshots of futures quotes. Algorithms~\ref{algo:macrocalib} and \ref{algo:spreadcalib} then calibrate a macroeconomic model and a parametric forward curve model, respectively, to historical data. Algorithm~\ref{algo:sim} summarizes the simulation procedure once the model has been calibrated. 

The overall modelling approach is modular and quite flexible. For example, the piecewise linear basis functions used to parameterize the forward curve in the numerical computations in this paper could easily be replaced by other functions. One could use e.g.~the classical Nelson-Siegel three-factor parameterization. Also, the Gaussian processes used to model the macroeconomic factors and the term structure risk factors could readily be replaced by other specifications. 

Instead of the macroeconomic model built in Algorithm~\ref{algo:macrocalib}, one could equally well plug in an alternative model that provides a description of the future development of the Fed rate $L$ that is used in step~6 of Algorithm~\ref{algo:sim} to transform the simulated process $x$ to the forward curve risk factors $\xi$. In particular, it would be interesting to study models where the forward curve has predictive power over the FOMC's future interest rate decisions. Such features could be incorporated in the framework of Section~\ref{sec:tsm} by regressing the macroeconomic variables $y$ on the forward curve risk factors $x$. Another interesting modification would be to regress the increments $\Delta y_t$ of the macroeconomic factors on their past values. Such a modification could add "inertia" to the macroeconomic variables to describe, for example, the tendency of the Fed to move interest rates in the same direction on consecutive FOMC meeting dates.

\begin{center}
\begin{minipage}{0.9\linewidth}
\begin{algorithm}[H]
\DontPrintSemicolon
\SetCustomAlgoRuledWidth{0.4pt}
  \KwIn{Daily time series for SOFR futures prices}
  \BlankLine
  \KwOut{Daily time series for term structure risk factors~$\xi_t=(\xi_t^k)_{k=1}^{K}$ (Figure~\ref{fig:xi})}
  \BlankLine
  Choose the $K$ tenors for the piecewise linear basis functions \eqref{eq:piecewise}\;
  \For{each day $t$}{
    Solve the calibration problem in \eqref{lsmid}
  }
 \caption{Construct the historical time series of the term structure risk factors~$\xi$ (Section~\ref{sec:par})}
\label{algo:estimatexi}
\end{algorithm}
\end{minipage}
\end{center}

\begin{center}
\begin{minipage}{0.9\linewidth}
\begin{algorithm}[H]
\DontPrintSemicolon
\SetCustomAlgoRuledWidth{0.4pt}
  \KwIn{Monthly historical time series of the macroeconomic variables $(L,I,G)$}
  \KwIn{Long-term medians of $(L,I,G)$ (see Table~\ref{table:ltmmacro})}
  \KwIn{Short-term median for process $L$ (The simulations of Section~\ref{sec:simexp} used the last observed forward curve computed using Algorithm~\ref{algo:estimatexi})}
  \BlankLine
 \KwOut{Autoregression matrix $A$, the constant vectors $a_t$ and the covariance matrix $\Sigma$ in equation~\eqref{eq:var} (see Tables~\ref{table:amatrixy},~\ref{table:corry})}
 \BlankLine
  Use the non-linear transformation \eqref{eq:macrotransform} to map the monthly values of $(L,I,G)$ to $y$\;
  Estimate the autoregression matrix $A$ in \eqref{eq:var} by applying linear regression to the three components of $y$\;
  Calculate an estimate of the covariance matrix $\Sigma$ using the residuals\;
  Calculate the constant vectors $a_t$ in equation~\eqref{eq:var} using the procedure of Section~4.2 of \cite{aap21} (see Section~\ref{sec:userviews})
 \caption{Calibrate the macroeconomic model (Sections~\ref{sec:datatr}, ~\ref{sec:macro}, ~\ref{sec:userviews})}
 \label{algo:macrocalib}
\end{algorithm}
\end{minipage}
\end{center}

\begin{center}
\begin{minipage}{0.9\linewidth}
\begin{algorithm}[H]
\DontPrintSemicolon
\SetCustomAlgoRuledWidth{0.4pt}
  \KwIn{Historical time series of the daily values of the term structure risk factors $\xi$ computed using Algorithm~\ref{algo:estimatexi}}
  \KwIn{Historical time series of monthly values of $L$}
  \KwIn{Long term medians of the forward curve risk factors $\xi$ (see Table~\ref{table:ltmxi})}
  \BlankLine
  \KwOut{Autoregression matrix $A$, the constant vectors $a_t$ and the covariance matrix $\Sigma$ in equation~\eqref{eq:varx} (see Tables~\ref{table:amatrixx},~\ref{table:corrx})}
  \BlankLine
  Use the nonlinear transformation given by~\eqref{eq:x_1} and~\eqref{eq:x_k} to map the daily values of  $(\xi^k)_{k=1}^K$ to $x$\;
  Estimate the autoregression matrix $A$ in~\eqref{eq:varx} by applying linear regression to the $K$ components of the time series $x$\;
  Calculate an estimate of the covariance matrix $\Sigma$ using the residuals\;
  Calculate the constant vectors $a_t$ in equation~\eqref{eq:varx} using the procedure of Section~4.2 of \cite{aap21} (see Section~\ref{sec:userviews})
 \caption{Calibrate the forward curve model (Sections~\ref{sec:datatr}, ~\ref{sec:termstruc}, ~\ref{sec:userviews})}
 \label{algo:spreadcalib}
\end{algorithm}
\end{minipage}
\end{center}

\begin{center}
\begin{minipage}{0.9\linewidth}
\begin{algorithm}[H]
\DontPrintSemicolon
\SetCustomAlgoRuledWidth{0.4pt}
  \KwIn{Autoregression matrices $A$, constant vectors $a_t$ and the covariance matrices $\Sigma$ from Algorithms~\ref{algo:macrocalib} and~\ref{algo:spreadcalib}}
  \KwIn{Starting values of $(L,I,G)$ and $\xi$}
  \KwIn{Start and end dates of the simulation}
  \KwIn{Number of scenarios to be simulated}
  \BlankLine
  \KwOut{Simulated time series for the macroeconomic variables $(L,I,G)$ (Figure~\ref{fig:simmacro})}
  \KwOut{Simulated time series of the forward curve $F$}
  \BlankLine
  Simulate Gaussian innovations $\epsilon_t$ of the macoeconomic model\;
  Calculate scenarios of the macroeconomic risk factors $y$ using equation \eqref{eq:var}\;
  Apply the inverse of the transformations \eqref{eq:macrotransform} to $y$ to get the scenarios for $(L,I,G)$\;
  Simulate Gaussian innovations $\epsilon_t$ of the forward curve model\;
  Calculate scenarios of the forward curve model $x$ using equation~\eqref{eq:varx}\;
  Apply the inverse of the transformations~\eqref{eq:x_1} and~\eqref{eq:x_k} to $x$ using simulated values of $L$ to get the scenarios for $\xi$\;
  Compute scenarios of the forward curve $F$ using expression~\eqref{eq:fcpar}
  
 \caption{Simulate the macroeconomic variables and the SOFR term structure (Section~\ref{sec:simexp})}
 \label{algo:sim}
\end{algorithm}
\end{minipage}
\end{center}

\bibliographystyle{plain}
\bibliography{sp}

\end{document}